\documentclass[11pt]{amsart}
\usepackage{geometry,verbatim,cite}
\usepackage{graphicx}
\usepackage{amssymb}
\usepackage{epstopdf}
\usepackage{amsmath,amscd}
\allowdisplaybreaks
\usepackage{mathrsfs}
\usepackage{caption}
\usepackage{subcaption}
\usepackage{color}
\usepackage{enumerate}
\usepackage{subcaption}
\usepackage{ulem}
\usepackage{url}

\newtheorem{theorem}{Theorem}[section]
\newtheorem{lemma}{Lemma}[section]

\newtheorem{definition}{Definition}[section]

\newtheorem{example}{Example}[section]
\newtheorem{assumption}{Assumption}


\newcommand{\MSspace}{\mathcal{X}}
\newcommand{\RR}{\mathbb{R}}
\newcommand{\R}{\mathcal{R}}
\newcommand{\Z}{\mathbb{Z}}
\newcommand{\M}{\rm{M}}
\newcommand{\N}{\mathbb{N}}
\newcommand{\C}{\mathbb{C}}
\newcommand{\CC}{\mathcal{C}}
\newcommand{\E}{\mathcal{E}}
\newcommand{\B}{\mathcal{B}}
\newcommand{\F}{\mathcal{F}}
\newcommand{\tK}{\tilde{K}}

\newcommand{\G}{\mathcal{G}}
\newcommand{\I}{\mathcal{I}}
\newcommand{\A}{\mathcal{A}}

\newcommand{\ms}{\text{ms}}
\newcommand{\Tr}{\text{Tr}}
\newcommand{\TrLim}{\underline{\Tr}}
\newcommand{\gauss}{\delta_\varepsilon}

\newcommand{\bw}{\tilde{\psi}}
\newcommand{\eH}{\widehat{H}}
\newcommand{\ER}{B_{\Sigma+\eta}}

\newcommand{\mBZ}{\Gamma_{\rm{M}}^*}
\newcommand{\mm}{\Theta}
\newcommand{\mR}{\R_{\rm{M}}^*}
\newcommand{\mG}{G_{\rm{M}}}
\newcommand{\iip}{\mathfrak{m}}
\newcommand{\ijp}{\mathfrak{n}}
\newcommand{\bthop}{\tilde{h}^{jj}_{\sigma\sigma'}}
\newcommand{\interh} {\hat{h}^{12}_{\sigma\sigma'}}
\newcommand{\GG}{{\widetilde{G}}}
\newcommand{\mE}{\widetilde{\mathcal{E}}}

\newcommand{\ex}{\widetilde{\chi}_r}

\newcommand{\eig}{\epsilon_i(q)}
\newcommand{\eigtp}{\epsilon_i(q,\iip,\ijp)}
\newcommand{\trcH}{\eH_r^{(\tau)}}
\newcommand{\tpH}{\eH_r^{(\iip,\ijp,\tau)}}
\newcommand{\op}{{\rm op}}
\newcommand{\htr}{d_\tau}
\newcommand{\bmH}{\widetilde{H}^{(\iip,\ijp,\tau)}}
\newcommand{\Jrl}[2]{J_{{#1} \leftrightarrow{#2}}}
\newcommand{\Rrl}[2]{R_{{#1} \leftrightarrow{#2}}}
\newcommand{\Hrl}[2]{H_{{#1} \leftrightarrow{#2}}}
\newcommand{\HT}{\mathbb{T}}
\newcommand{\Gjm}{\mathfrak{G}_{\M}}
\newcommand{\Gmj}{\mathfrak{G}_{j}}
\newcommand{\cP}{\widetilde{P}}
\newcommand{\cU}{\widetilde{U}}
\newcommand{\LS}{L_{\rm BZ}}

\title[Incommensurate Continuum Models]{Construction and Accuracy of Electronic Continuum Models of Incommensurate Bilayer 2D Materials}
\author{Xue Quan}
\address[Xue Quan]{Beijing Normal University, Beijing, China}
\email{xuequan@mail.bnu.edu.cn}
\author{Alexander B. Watson}
\address[Alexander B. Watson]{University of Minnesota, Minneapolis, Minnesota, U.S.A.}
\email{abwatson@umn.edu}
\author{Daniel Massatt}
\address[Daniel Massatt]{Louisiana State University, Baton Rouge, Louisiana, U.S.A.}
\email{dmassatt@lsu.edu}

\date{\today}

\begin{document}

\maketitle






\begin{abstract}
    Single-particle continuum models such as the popular Bistritzer-MacDonald model have become powerful tools for predicting electronic phenomena of incommensurate 2D materials and developing many-body models aimed at modeling unconventional superconductivity and correlated insulators. In this work, we introduce a procedure to construct continuum models of arbitrary accuracy relative to tight-binding models for moir\'{e} incommensurate bilayers. This is done by recognizing the continuum model as arising from Taylor expansions of a high accuracy momentum space approximation of the tight-binding model. We apply our procedure in full detail to two models of twisted bilayer graphene and demonstrate both admit similar Bistritzer-MacDonald models as the leading order continuum model, while higher order expansions reveal qualitative spectral differences.
    
\end{abstract}

\section{Introduction}

\begin{figure}[htb!]
\centering
\includegraphics[scale=.165]{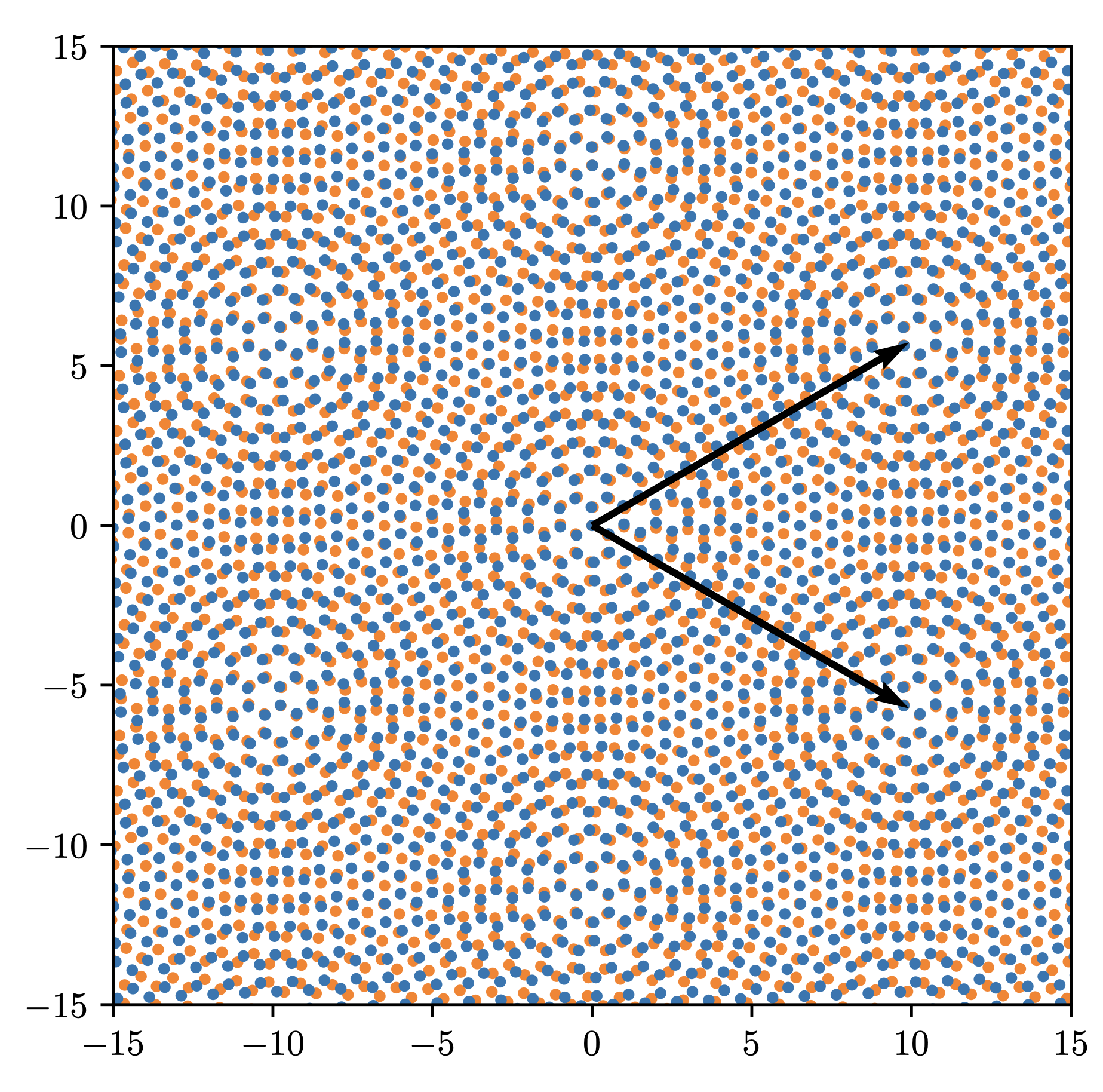}  
\includegraphics[scale=.26]{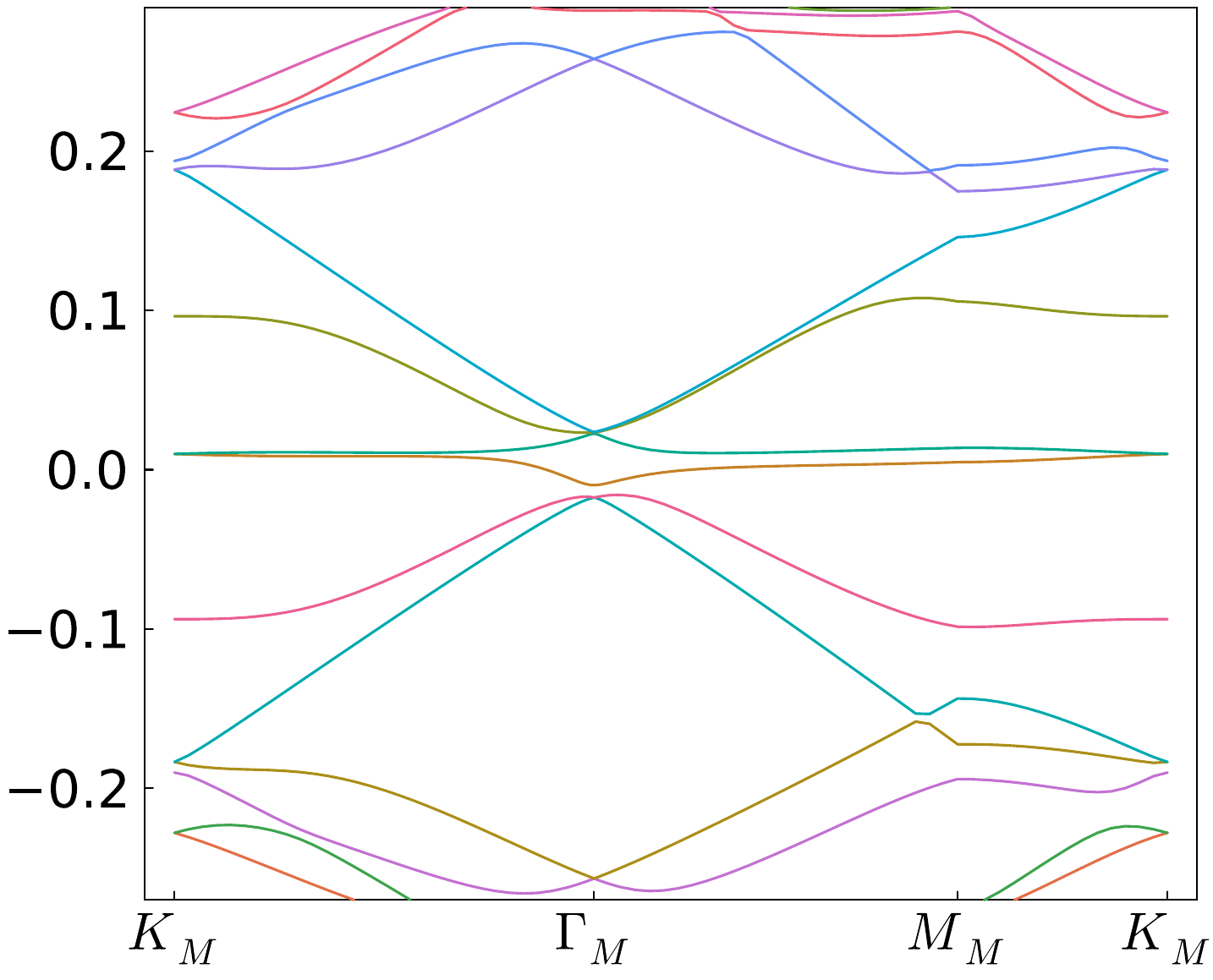}
\caption{
Left: atomic structure of stacked bilayer graphene with interlayer twist $5^\circ$, viewed from above. Atoms in the lower graphene sheet are colored orange and in the upper graphene sheet in blue. Moir\'e lattice vectors which generate the moir\'e lattice are shown in black. The continuum models we construct in the present work have coefficients which are periodic with respect to this lattice. Right: band structure of the Bistritzer-MacDonald continuum model \cite{Bistritzer2011} at the magic twist angle $1.1^\circ$, with parameter values taken from the ab initio tight-binding model proposed in \cite{Fang2016}. For more details, see Figure \ref{fig:bandstruct}; this figure is the same as the middle right figure there. Many-body models of twisted bilayer graphene typically model electrons in the flat bands, the two bands nearest to zero energy which are approximately flat over most of the moir\'e Brillouin zone, interacting via the Coulomb force. Two of the key inputs to such models are the flat band eigenvalues and associated eigenfunctions. The goal of the present work is to construct moir\'e-scale single-particle continuum models of twisted bilayer graphene's electronic properties with essentially arbitrary accuracy relative to tight-binding models, in order to provide more accurate flat band eigenvalues and eigenfunctions for many-body models.
}
\label{fig:TBG_struc}
\end{figure}

Research on incommensurate stackings of 2D materials has grown rapidly since the discovery of correlated insulator and ``unconventional'' superconducting states in twisted bilayer graphene (TBG) at the ``magic angle'' $\approx 1.1^\circ$ in 2018 \cite{Cao2018,Cao2018a,kim2021,ke2021}. To understand these phenomena, which arise from electron-electron correlations, there is a strong push to construct accurate yet tractable many-body models of such materials \cite{Bernevig2021, Romanova2022, Faulstich2023}. Many-body models are generally constructed by the following recipe. First, construct an accurate model for the single-particle electronic properties of the material. Then, extract a basis of energetically relevant states from this model. Finally, project the electron-electron Coulomb interaction to this basis. Hence a careful derivation of an accurate single-particle basis is fundamental.

A key feature of incommensurate stackings of 2D materials at relatively small twist angles is the emergence of approximate periodicity in the atomic structure known as the moir\'e pattern (see Figure \ref{fig:TBG_struc}). The moir\'e pattern allows for constructing relatively simple single-particle models which are, in particular, periodic with respect to the bilayer moir\'e pattern. 
Such models admit a band structure: their spectrum is the union of intervals (bands) swept out by real functions $k \mapsto E_n(k)$, where $n \in \mathbb{Z}$ is known as the band index, and $k$ varies over the moir\'e Brillouin zone. The associated eigenfunctions of these eigenvalues are Bloch functions $e^{i k \cdot r} u_n(r;k)$, where $u_n$ is moir\'e-periodic. When constructing many-body models, the energetically relevant single-particle states are usually taken to be the set of Bloch functions associated to the single-particle bands near to the Fermi level.


The most famous moir\'e-periodic single-particle model for twisted bilayer graphene is the Bistritzer-MacDonald (BM) model \cite{Bistritzer2011}. This model comprises a system of continuum Dirac equations, characteristic of the electronic structure of monolayer graphene, coupled through a matrix-valued moir\'e-periodic potential modeling interlayer hopping depending on the local stacking configuration. This model has been mathematically rigorously justified as an accurate model for electron dynamics under assumptions which are expected to hold in twisted bilayer graphene at the magic angle \cite{watson2023}: specifically, that the twist angle $\theta \ll 1$ (in radians), and that the typical interlayer hopping energy is much less than the typical monolayer hopping energy. For other mathematical derivations, see \cite{quinn2025higherordercontinuummodelstwisted,cances2023,cances2023semiclassical,cances2021secondorder}. The Bistritzer-MacDonald model played an important role in identifying TBG's magic angles, twist angles where the BM model's Bloch bands near to the Fermi level become relatively flat. This phenomenon is thought to be the reason why magic angle TBG exhibits such a rich many-body phase diagram.

The results of the present work are as follows. First, we introduce a systematic procedure for constructing moir\'e-periodic single-particle continuum models of TBG in the same regime where the BM model is valid, but to essentially arbitrary accuracy, assuming the validity of an atomic-scale tight-binding model. Second, we rigorously prove convergence of the density of states distribution of the underlying tight-binding model to that of the continuum models we construct. Finally, we compute band structures for the models we construct, comparing the results starting from a simplified tight-binding model as in \cite{watson2023} and a more elaborate model parametrized by density functional theory (DFT) computations \cite{Fang2018}. We find that both models admit similar Bistritzer-MacDonald models as their leading-order approximations, but higher-order expansions reveal qualitative differences. 

The models we construct account for terms which are neglected in the Bistritzer-MacDonald model, such as higher-order dispersion in the intralayer part of the Hamiltonian, and longer range momentum space hopping and derivative terms in the interlayer part. We emphasize that the energies of these terms are comparable to energy differences between competing many-body ground states in interacting models of twisted bilayer graphene \cite{Faulstich2023}, suggesting they may be important in correctly predicting TBG's many-body phase diagram. 

The construction works by building on efficient methods for accurately computing the electronic structure of incommensurate two-dimensional multi-layered materials \cite{Massatt2017,MassattCarrLuskinOrtner2018,Massatt2020,Massatt2023,PhysRevResearch.2.033162,Wang2025}. The idea is as follows. First, as in \cite{MassattCarrLuskinOrtner2018,Massatt2023}, we identify the set of energetically relevant electronic states within each layer; in twisted bilayer graphene these are the states near to the monolayer Dirac points. We then truncate the full tight-binding model to these states, accruing exponentially small error in the truncation threshold as proved in \cite{MassattCarrLuskinOrtner2018}. 
We then further simplify the truncated model by Taylor-expanding about the Dirac points of each layer. Finally, we realize our accurate models as natural extensions of the truncated model to arbitrary wave-numbers, i.e., taking the continuum limit. 

By controlling the error at each step, we obtain error estimates in approximating the density of states distribution in terms of three truncation parameters: the intralayer Taylor-expansion order, interlayer expansion order, and momentum-space hopping distance. In the limit where these parameters become large, we recover the exponential rate of convergence of \cite{MassattCarrLuskinOrtner2018}. These results are summarized in our main result, Theorem \ref{thm:convergence:dos}. This approach can be easily generalized to other two-dimensional material stackings such as layered transition metal dichalcogonides and incommensurate trilayer systems, for which moir\'e-periodic single particle models analogous to the Bistritzer-MacDonald model have been proposed \cite{tmdc2021, koshino2023}.

We finally discuss related mathematical work. In \cite{quinn2025higherordercontinuummodelstwisted}, higher-order corrections to the Bistritzer-MacDonald model have been derived following a different methodology from the present work, and proving different results. Specifically, in that work error estimates are proved for the dynamics generated by a tight-binding Schr\"odinger equation of twisted bilayer graphene by a multiple-scales approach. In \cite{cances2023,cances2023semiclassical,cances2021secondorder}, moir\'e-scale continuum models such as the BM model have been derived as approximations of a DFT-informed incommensurate linear Schr\"{o}dinger model by making use of a semiclassical expansion, with accuracy quantified for the density of states. In \cite{watson2021,BeckerEmbreeWittstenZworski2021,BeckerEmbreeWittstenZworski2022,Humbert2022} insight has been gained into the Bistritzer-MacDonald model's almost flat bands by studying an approximation of this BM model known as the chiral model introduced in \cite{Tarnopolsky2019}. The chiral model also has an important role in studying TBG's many-body electronic properties. Many-body models for TBG's electronic properties have been proposed and solved in \cite{Faulstich2023,becker2023exactgroundstateinteracting,stubbs2024hartreefockgroundstatemanifold}, although these studies generally rely on the Bistritzer-MacDonald model without the corrections we derive here. Another important aspect of the physics of twisted bilayer graphene near the magic angle are the material's mechanical/structural properties. The ground state (relaxed) mechanical properties have been investigated in \cite{2018CarrMassattTorrisiCazeauxLuskinKaxiras,Cazeaux2020,Hott2024}, while phonons have been recently considered in \cite{hott2024mathematicalfoundationsphononsincommensurate}.

\subsection{Structure of paper} The structure of the rest of this paper is as follows. In Section 2, we introduce the momentum space formulation as a transformation of real space tight-binding models and review previous momentum space results we need. In Section 3, we derive continuum models as approximations of momentum space models, and analyze the convergence of the band structure. In Section 4, we verify the convergence results for the simplified TBG and the ab initio Wannierized TBG, and derive high accuracy continuum models for each. The proofs are in the Appendices for clarity of presentation.

\subsection{Acknowledgements}

ABW's research was supported in part by NSF grant DMS-2406981.
DM's research was supported in part by AFRL grant FA9550-24-1-0177.
XQ's research was supported by the National Natural Science Foundation of China (No. 12371431).

\section{Tight-Binding Model}

In this work we consider two tight-binding models of twisted bilayer graphene (TBG): the relatively simple model considered in \cite{watson2023}, and the more elaborate model proposed in \cite{Fang2016} based on ab initio density functional theory calculations. In this section we describe a general tight-binding model of TBG and explain how to obtain the simplified model as a special case of this model.


\subsection{Monolayer graphene model}

 We start by describing a general tight-binding model for an unrotated graphene monolayer. Since graphene is a sheet of carbon atoms arranged in a honeycomb lattice, we introduce Bravais lattice vectors as
\[
a_1 := \frac{a}{2}(1,\sqrt{3})^T, \qquad a_2 := \frac{a}{2}(-1,\sqrt{3})^T, \qquad A := (a_1, a_2),
\]
where $a>0$ is the lattice constant.
The graphene Bravais lattice and its unit cell are defined by
\begin{equation*}
    \R := A\Z^2, \qquad \Gamma := A[0,1)^2.
\end{equation*}
The associated reciprocal lattice and reciprocal lattice unit cell (Brillouin zone) are given by
\begin{equation*}
    \R^* := 2\pi A^{-T}\Z^2, \qquad \Gamma^* := 2\pi A^{-T}[0,1)^2.
\end{equation*}
The Dirac points, high-symmetry points in the Brillouin zone, play a critical role in the electronic properties of graphene and are given by
\begin{equation*}
    K := \frac{4\pi}{3a}(1,0)^{T}, \qquad K':=-K.
\end{equation*}
We further define orbital displacement vectors (with $d := \frac{a}{\sqrt{3}}$):
\begin{equation*}
    \tau_A := (0,0)^T, \qquad \tau_B := (0,d)^T,
\end{equation*}
so that the two atoms in the $R$-th lattice are located at $R+\tau_A$ and $R+\tau_B$.
Then the real space Hamiltonian can be defined via the infinite matrix
\[
H_{R\sigma,R'\sigma'}:=h_{\sigma\sigma'}(R-R'),
\]
where $h_{\sigma\sigma'}:\R\to\C,h_{\sigma\sigma'}\in\ell^1(\R)$.
For $\psi = (\psi_R)_{R\in\R}=(\psi_{RA},\psi_{RB})^T_{R\in\R}\in \ell^2(\R \times \{A,B\})$, $H$ acts as
\begin{equation}
\label{MGham}
(H\psi)_{R\sigma} = \sum_{\sigma'\in\{A,B\}}\sum_{R'\in\R}h_{\sigma\sigma'}(R-R')\psi_{R' \sigma'}.
\end{equation}

\subsection{Twisted bilayer tight-binding model}

Twisted bilayer graphene (TBG) consists of two graphene monolayers with a relative twist. We let 
\[
R(\vartheta) := \begin{pmatrix} \cos \vartheta & -\sin \vartheta \\ \sin \vartheta & \cos\vartheta\end{pmatrix}
\]
be the rotation matrix and $\theta>0$ be the twist angle. The lattice vectors of each layer are defined by
\begin{align*}
    a_{1,j} := R\biggl(-\frac{\theta}{2}\biggr)a_j, \qquad a_{2,j} := R\biggl(\frac{\theta}{2}\biggr)a_j, \qquad A_j := (a_{j,1}, a_{j,2} ).
\end{align*}
We then have
\[\R_j := A_j\Z^2, \qquad\Gamma_j := A_j[0,1)^2,\]
\[\R_j^* := 2\pi A_j^{-T}\Z^2,\qquad\Gamma_j^* := 2\pi A_j^{-T}[0,1)^2.\]
We require the following definition and assumption.
\begin{definition}
    Two Bravais lattices $\mathcal{L}_1$ and $\mathcal{L}_2$ are incommensurate if 
    \[
    \mathcal{L}_1 \cup \mathcal{L}_2 + v = \mathcal{L}_1 \cup \mathcal{L}_2 \qquad \Leftrightarrow \qquad v = \begin{pmatrix} 0 \\ 0\end{pmatrix}.
    \]
\end{definition}
\begin{assumption} \label{as:incommensurate}
We assume from this point on that the lattices $\R_1$ and $\R_2$ are {\em incommensurate}, and that the dual lattices $\mathcal{R}_1^*$ and $\mathcal{R}_2^*$ are also incommensurate.
\end{assumption}
Note that incommensurability of the dual lattices is not equivalent to incommensurability of the real space lattice; see, e.g., \cite{Wang2025}.

The Dirac points and orbital displacements are similarly rotated
\begin{align*}
K_1:=R\biggl(-\frac{\theta}{2}\biggr)K, \qquad K_2:=R\biggl(\frac{\theta}{2}\biggr)K,\qquad K'_j:=-K_j,
\end{align*}

\begin{align*}
    \tau_{1\sigma} := R\biggl(-\frac{\theta}{2}\biggr)\tau_\sigma, \qquad \tau_{2 \sigma} := R\biggl(\frac{\theta}{2}\biggr)\tau_\sigma .
\end{align*}
Let $\A_1:=\{1A,1B\}$ and $\A_2:=\{2A,2B\}$ be the set of indices of orbital displacements for sheet $1$ and $2$ respectively. Let $\Omega_j:=\R_j\times\A_j$, then the full degree of freedom space is $\Omega := \Omega_1\cup\Omega_2$. 
We write the wave functions in TBG by $\psi = (\psi_1,\psi_2)^T\in\oplus_{j=1}^2\ell^2(\Omega_j)$, the incommensurate tight-binding model acts as 
\[
H\psi = \begin{pmatrix} H_{11} & H_{12} \\ H_{12}^\dagger & H_{22}\end{pmatrix}\begin{pmatrix} \psi_1 \\ \psi_2\end{pmatrix}.
\]
Here the intralayer part $H_{jj}$ is similarly given by
\[
(H_{jj})_{R\sigma,R'\sigma'} := h^{jj}_{\sigma\sigma'}(R-R'),\qquad R\sigma,R'\sigma'\in\Omega_j,\qquad j\in\{1,2\},
 \]
with $h^{jj}_{\sigma\sigma'}:\R_j\to\C,h^{jj}_{\sigma\sigma'}\in\ell^1(\R_j)$, and the interlayer part $H_{12}$ is given by 
\[
(H_{12})_{R\sigma,R'\sigma'} := h^{12}_{\sigma\sigma'}(R+\tau_{\sigma}  - R'-\tau_{\sigma'}), \qquad R\sigma\in\Omega_1,~ R'\sigma'\in\Omega_2,
\]
with $h^{12}_{\sigma\sigma'}:\RR^2\to\C,h^{12}_{\sigma\sigma'}\in C(\RR^2)$. We denote the Fourier transform pair by 
\begin{equation}
\label{fourier}
\begin{split}
    &\F h(\xi):=\hat{h}(\xi) := \int_{\RR^2}e^{-i x\cdot \xi} {h}(x)\; dx \\
    &\F^{-1}\hat{h}(x):=\frac{1}{(2\pi)^2}\int_{\RR^2}e^{i x\cdot \xi} \hat{h}(\xi)\; d\xi.
\end{split}
\end{equation}
We will employ the following assumption for interlayer hopping functions:
\begin{assumption}
    \label{assumption:interhop}
    The interlayer hopping functions $h^{12}_{\sigma\sigma'}\in C(\RR^2)$ for $\sigma\in\A_1,\sigma'\in\A_2$ satisfy for some $\gamma_{12},\gamma_{12}'>0$
    \[
|h^{12}_{\sigma\sigma'}(x)|\lesssim e^{-\gamma_{12}|x|} \quad \forall x \in \mathbb{R}^2 \qquad {\rm and}\qquad|\hat{h}^{12}_{\sigma\sigma'}(\xi)|\lesssim e^{-\gamma_{12}'|\xi|} \quad \forall \xi \in \mathbb{R}^2,
\]
where we write $a \lesssim b$ to denote that there exists a constant $C > 0$ independent of parameters appearing on either the left or right-hand side so that $a \leq C b$.
\end{assumption}
For intralayer, we use the following assumption:
\begin{assumption}
\label{assumption:intrahop}
    The intralayer hopping functions $h^{jj}_{\sigma\sigma'}\in\ell^1(\R_j)$ for $\sigma,\sigma' \in \A_j$ satisfy for some $\gamma_j > 0$
    \[
    |h^{jj}_{\sigma\sigma'}(R)| \lesssim e^{-\gamma_j |R|} \quad \forall R \in \mathcal{R}_j.
    \]
\end{assumption}

We next describe the simplified model as a special case of the above structure. The intralayer and interlayer hopping functions arise in practice from overlap integrals of exponentially-decaying Wannier orbitals. A physically reasonable choice of such functions is therefore
\begin{align} \label{eq:toy_hs}
    &h^{jj}_{\sigma \sigma'}(x) = e^{- \alpha |x + \tau_{\sigma} - \tau_{\sigma'}|}, \qquad j \in \{1,2\}, 
    \\[1ex]\label{toy_hs_inter}
    &h^{12}_{\sigma \sigma'}(x) = h_{12}(x)= e^{- \beta \sqrt{ |x|^2 + z^2 }},
\end{align}
for constants $\alpha, \beta, z > 0$. Here, $z$ models the non-zero interlayer displacement. These functions clearly satisfy Assumptions \ref{assumption:interhop} and \ref{assumption:intrahop}; the Fourier transform of the interlayer hopping function is \cite{Bateman1954a}
\begin{equation} \label{eq:toy_hs_FT}
    \hat{h}_{12}(\xi) = 2\pi\frac{\beta e^{- z \sqrt{ |\xi|^2 + \beta^2 } } \left( 1 + z \sqrt{ |\xi|^2 + \beta^2 } \right) }{ ( |\xi|^2 + \beta^2 )^{3/2} }.
\end{equation}
More accurate forms for the functions \eqref{eq:toy_hs}-\eqref{eq:toy_hs_FT} can be obtained from DFT calculations; this is the approach of \cite{Fang2016}, which we refer to as the Wannierized model. We provide computations using this model below.

\subsection{Momentum Space}
\label{subsec:ms}
We now transform the real space Hamiltonian to the momentum space setting. We denote wave functions in the Bloch domain $L^2(\Gamma_j^*;\C^2)$ by $\tilde{\psi}_j=\big(\tilde{\psi}_j(q)\big)_{q\in\Gamma_j^*}=\big((\tilde{\psi}_{j})_{jA}(q),(\tilde{\psi}_{j})_{jB}(q)\big)^T_{q\in\Gamma_j^*}$, and define unitary Bloch transforms in each layer by
\begin{equation}
\begin{split}
    &[{\mathcal{G}}_j {\psi}_j]_{\sigma}(q) := \frac{1}{|\Gamma_j^*|^{1/2}}\sum_{R \in \R_j} e^{-iq\cdot (R + \tau_{\sigma})}(\psi_{j})_{R \sigma} \\
    &[{\mathcal{G}}_j^*{\tilde{\psi}}_j]_{R\sigma} := \frac{1}{|\Gamma_j^*|^{1/2}}\int_{\Gamma_j^*}e^{iq\cdot (R + \tau_{\sigma})}\tilde{\psi}_{j\sigma}(q)\;d q,
\end{split}
\end{equation}
for $R\sigma\in\Omega_j$, $j \in \{1,2\}$.
Note that with this convention, Bloch transforms are quasi-periodic with respect to $\R_j^*$ in the sense that  
\begin{equation}
[\G_j\psi_j]_\sigma(q+G) = e^{-iG\cdot\tau_{\sigma}}[\G\psi_j]_\sigma(q), \qquad G\sigma\in\R_j^*\times\A_j,
\end{equation}
so the functions in $L^2(\Gamma_j^*;\C^2)$ can be defined over $\RR^2$.
We denote momentum space bilayer wavefunctions by
\begin{equation}
    \tilde{\psi} = ( \tilde{\psi}_1 , \tilde{\psi}_2 )^T \in \oplus_{j=1}^2 L^2(\Gamma_j^*;\mathbb{C}^2).
\end{equation}
Let $\G = (\G_1,\G_2)$, the momentum space Hamiltonian is defined via
\begin{equation}
H^\ms\bw := \G H\G^*\bw
 = \begin{pmatrix} \G_1H_{11}\G_1^*\bw_1 +\G_1H_{12}\G_2^*\bw_2 \\[1.5ex]
 \G_2H_{12}^\dagger \G_1^*\bw_1 +\G_2H_{22}\G_2^*\bw_2\end{pmatrix}.
\end{equation}
Up to a phase change in the definition of the Bloch transform, the calculation from previous work \cite{MassattCarrLuskinOrtner2018} gives
\begin{equation}
\begin{split}
     &[\G_jH_{jj}\G_j^*]_ {\sigma\sigma'} =  \bthop(\cdot) :=  \sum_{R_j \in \mathcal{R}_j} e^{- i (\cdot) \cdot (R_j + \tau_{\sigma} - \tau_{\sigma'})} h_{\sigma \sigma'}^{jj}(R_j),\qquad \sigma,\sigma'\in\A_j
\\
     &[\G_1H_{12}\G_2^*]_{\sigma\sigma'} =c_1c_2\sum_{G \in \R_1^*} e^{iG\cdot\tau_{\sigma}}\interh(\cdot+G) T_G,\qquad\sigma\in\A_1,\sigma'\in\A_2
\end{split}  
\end{equation}
where $c_j:=|\Gamma_j|^{-1/2}$, 
and $T_G$ is the translation operator on $L^2(\Gamma_j^*;\C^2)$:
\[
(T_G\bw_j)(q) = \bw_j(q+G).
\]

Recalling that we assume incommensurability of the bilayer real space lattices (Assumption \ref{as:incommensurate}), using Birkhoff's ergodic theorem, a large class of observables associated to the momentum space Hamiltonian $H^\text{ms}$ equal the same observables associated to a discrete ``reciprocal space'' Hamiltonian $\widehat{H}$ describing hopping on the bilayer reciprocal lattices. This perspective turns out to be the most convenient for justifying truncation of the Hamiltonian to a family of finite-dimensional matrices for numerical computation. In deriving the reciprocal space Hamiltonian we have to restrict attention to the class of analytic functions in momentum space, rather than considering the full $L^2$ space, so that we can consider pointwise values of the momentum space wavefunctions. Since when we evaluate observables we consider the trace over plane wave basis functions, which clearly belong to this space, this restriction is not significant.

We now derive the reciprocal space Hamiltonian following \cite{Massatt2023,MassattCarrLuskinOrtner2018}. Let $\MSspace_j$ be the space of analytic periodic functions in $L^2(\Gamma_j^*;\C^2)$, and $\Omega_j^*:=\R_{{\rm F}_j}^* \times \A_j$ with ${\rm F}_1 := 2, {\rm F}_2 :=1$.
We define a map $\E_{q,j}:\MSspace_j\to\ell^\infty(\Omega_j^*)$ that takes functions defined over momenta to a discrete dense sampling of the function 
\begin{align*}
    & \{\E _{q,j}\bw_j\}_{G\sigma}:=\bw_{j\sigma}(q+G).
\end{align*}
Let $\Omega^* =\Omega_1^*\cup\Omega_2^*$ and
\[
\E_q := \begin{pmatrix} \E_{q,1}& 0 \\ 0 & \E_{q,2}\end{pmatrix},
\]
the reciprocal space Hamiltonian $\widehat{H}(q) : \ell^2(\Omega^*)\rightarrow \ell^2(\Omega^*)$ takes the form
\begin{align}
\label{unfold}
\E_q \big(H^\ms\bw\big)=:\eH(q)\E_q\bw,\qquad \bw \in \oplus_{j=1}^2 \MSspace_j.
\end{align}
Likewise, $\eH(q)$ can be denoted in a sheet-wise decomposition
\begin{equation}
\label{recipH}
    \eH(q) = \begin{pmatrix} \eH_{11}(q)& \eH_{12}(q) \\ {\eH_{12}}^\dagger(q) &\eH_{22}(q) \end{pmatrix},
\end{equation}
with intralayer entries  
\begin{align*}
[\eH_{jj}(q)]_{G,G'} :=  \tilde{h}^{jj}(q+G)\delta_{GG'}, \qquad G,G'\in\R_{{\rm F}_j}^*
\end{align*}
and interlayer entries
\begin{align*}
[\eH_{12}(q)]_{G,G'} := c_1c_2\HT_{G,G'}\odot{\hat{h}^{12}}(q+G+G'), \qquad G\in \R_2^*, \;G' \in \R_1^*.
\end{align*}
Here $\HT_{G,G'}^{\sigma\sigma'} := e^{iG'\cdot\tau_{\sigma}} e^{-iG\cdot\tau_{\sigma'}}$ is the phase matrix, and $\odot$ is the Hadamard product operator so that for $2\times 2$ matrices $\HT,\mathbb{S}$ and $\sigma\in\A_1,\sigma'\in\A_2$,
\[
(\HT\odot\mathbb{S})_{\sigma\sigma'}:=(\HT)_{\sigma\sigma'}(\mathbb{S})_{\sigma\sigma'}.
\]

We further express electronic observables via \eqref{recipH}.
The density of states (DoS) is approximated by the Gaussian smearing for the thermodynamic limit trace 
\begin{equation} \label{eq:full_DoS}
D_\varepsilon(E):=\TrLim \; \gauss(E-H):= \lim_{r\to\infty}\frac{1}{\#\Omega_r}\sum_{R\sigma\in\Omega_r} [\delta_\varepsilon(E-H)]_{R\sigma,R\sigma},
\end{equation}
where $\delta_\varepsilon(E-\lambda):=\frac{1}{\varepsilon\sqrt{2\pi}}e^{-(E-\lambda)^2/2\varepsilon^2}$ is the Gaussian centered at $E\in\RR$, and $\Omega_r:=\{R\sigma\in\Omega:|R|<r\}$ is the restricted space with $r>0$.
The DoS can be re-expressed using an ergodic theorem from \cite{MassattCarrLuskinOrtner2018}, which we state for reference: 
 \begin{theorem}
     Assume Assumption \ref{as:incommensurate} and assume $H$ has hopping functions satisfying Assumptions \ref{assumption:interhop} and \ref{assumption:intrahop}. Then we have
\begin{equation}
\label{DOS:ms}
   D_\varepsilon(E) = \nu^* \sum_{j=1}^2\sum_{\sigma \in\A_j} \int_{\Gamma_j^*} [\delta_\varepsilon(E-\widehat{H}(q))]_{0\sigma,0
\sigma}\;dq
\end{equation}
with 
\[
\nu^* = \frac{1}{2} \left(\sum_{j=1}^2|\Gamma_j^*|\right)^{-1}.
\]
 \end{theorem}

\subsection{Reciprocal Space Truncation}

We next review the algorithm that truncates the reciprocal space formulation \eqref{recipH} as a finite sparse matrix to make the problem computationally tractable and numerically efficient \cite{MassattCarrLuskinOrtner2018,Massatt2023}.
When we are only interested in a specific range of energies, such as a small subset of the monolayer spectra, we can select a finite basis set based on the energy and momenta relationship arising from the monolayer band structure. 
Specifically, we define for an energy region $\mathrm{E}\subset\RR$ and $r>0$, the 
corresponding momenta regions in $\Gamma_j^*$
\begin{equation*}
    \Gamma_j^*(\mathrm{E}):=\Big\{q\in\Gamma_j^*:\sigma_j(q)\cap \mathrm{E}\neq\emptyset\Big\},
\end{equation*}
where $\sigma_j(q)$ is the set of eigenvalues of $\eH_{jj}(q)$, and the corresponding set of basis elements
\[
\Omega_r^*(q,\mathrm{E}):=\Big\{G\sigma\in\Omega_j^*:{\rm{mod}}_j(q+G)\in\Gamma_j^*(\mathrm{E})+B_r(0),\,j=1,2\Big\},
\]
where $B_r(0)$ is a ball in $\RR^2$ of radius $r$ centered at $0$, set addition is defined as $U+V :=\{u+v : u \in U, \; v \in V\}$, and $\text{mod}_j$ denotes values modulo the reciprocal lattice of layer $j$.
The set $\Omega_r^*(q,\mathrm{E})$ will generally be infinite because of incommensurability. 

Here we are interested in the case of narrow energy regions centered at the Dirac energy in graphene, and small twist angles. In this case, for sufficiently small $r > 0$,
$\Omega_r^*$ decomposes into finite subsets centered at the $K$ or $K'$ Dirac points, such that the distance between these subsets and their ``replicas'' arising from the modular arithmetic is proportional to inverse twist angle. This happens because, whenever $G$ is a reciprocal lattice vector of one layer, the value of $G$ modulo the other layer's reciprocal lattice has length proportional to $\theta$.
In order to pick out the region containing the starting momentum $q$, we define the map $I:\B(\Omega^*)\to\B(\Omega^*)$ as the operation that maps a subset of $\Omega^*$ to its isolated degrees of freedom containing the site $0\sigma$ for any $\sigma$, where $\B(\Omega^*)$ is the collection of subsets of $\Omega^*$. 

Assume we are only interested in the energy region around Dirac points $K,\;K'$, i.e., $B_\Sigma := (-\Sigma,\Sigma)$, $\Sigma > 0$. We expand it by $\ER:=(-\Sigma-\eta,\Sigma+\eta)$, with $\eta$ is a little larger than twice the interlayer coupling strength,
\begin{equation}
    \label{couple_strength}
    \eta :=(2+\alpha) \sup_{\mathop{\psi\in\ell^2(\Omega^*)}}\frac{1}{\|\psi^{(1)}\|_2\|\psi^{(2)}\|_2}(\psi^{(1)},0)\eH(q)
    \begin{pmatrix}
        0\\ \psi^{(2)}
    \end{pmatrix},
\end{equation}
where $\alpha > 0$. The basis space corresponding to $B_\Sigma$ is 
\begin{align}
\label{finite basis}
    \Omega_r^*(q) := I\big(\Omega_{r}^*(q,\ER)\big).
\end{align}
Note that for $r=0$, the momenta regions are bounded by the balls $B_{r_{\Sigma}}\big({\rm mod}_j(K_j)\big)$ and $B_{r_{\Sigma}}\big({\rm mod}_j(K_j')\big)$ with
\begin{equation}
    r_\Sigma:= \max_{\substack{q\in\Gamma_j^*(\ER)\\j=1,2}}\min_{\tK_j\in\{K_j,K_j'\}}\big|q-{\rm mod}_j(\tK_j)\big|.
\end{equation}
As an example for the relationship between energy, momenta regions and bounding balls, see Figure \ref{fig:momenta_trunc}.
\begin{figure}[htb!]
\centering
\includegraphics[height=6.5cm]{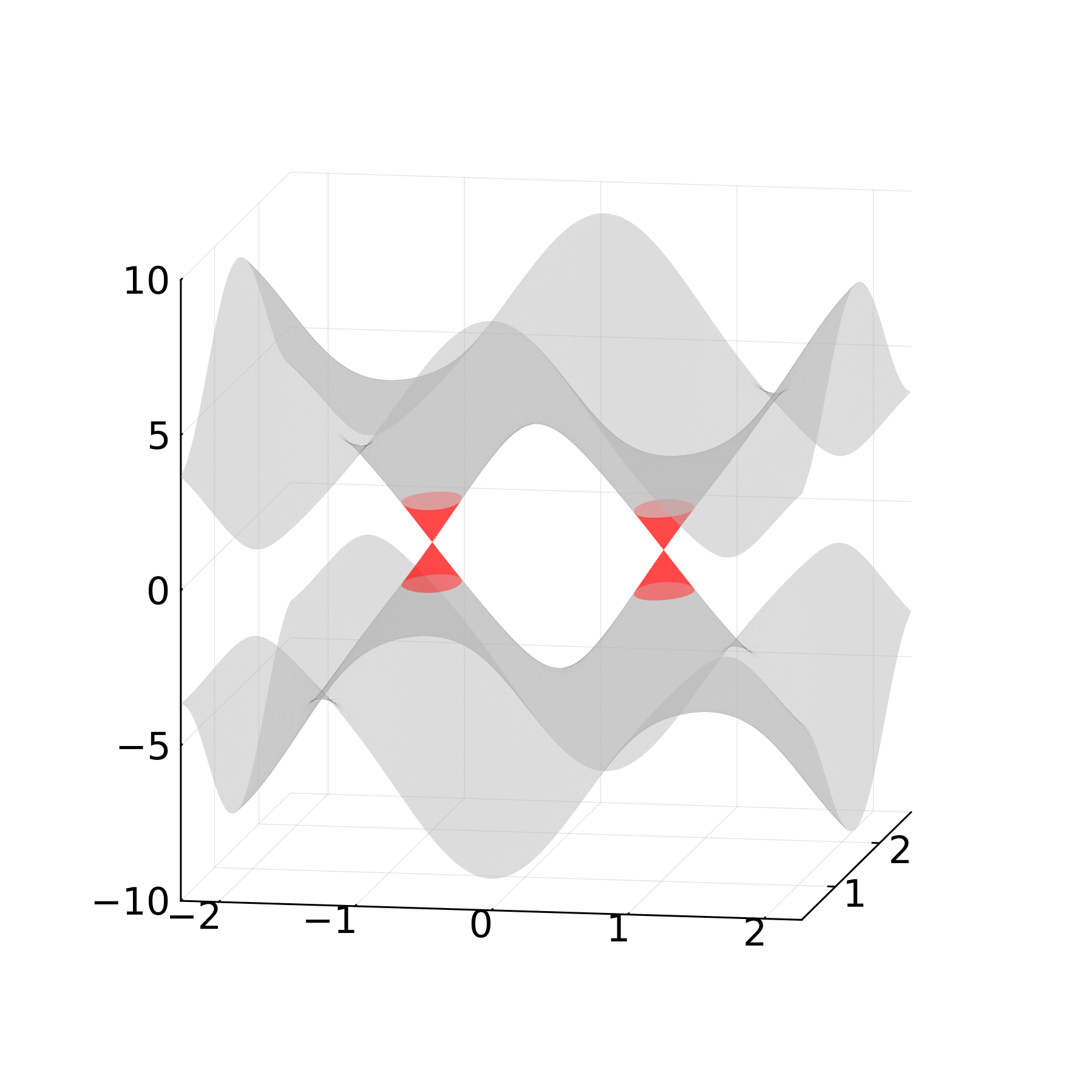} 
\includegraphics[height=5.5cm]{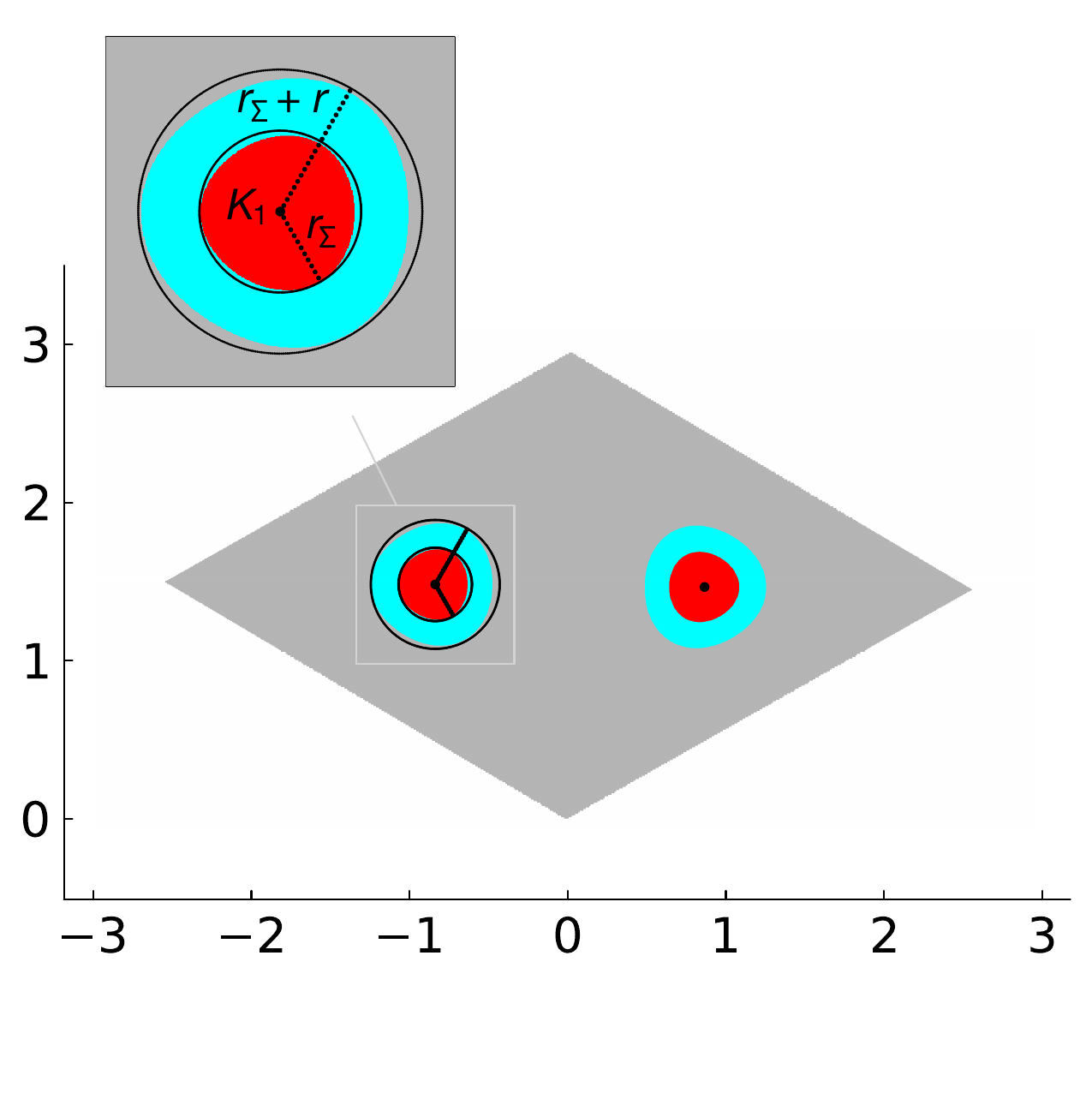}  
 
\setlength{\abovecaptionskip}{0pt} 
\caption{
The left picture shows the rotated monolayer graphene band structure, where the red regions correspond to ${\ER}$. The right picture shows the reciprocal unit cell of the rotated monolayer graphene, where the red regions are $\Gamma_1^*(\ER)$, and by adding the blue regions around them we obtain $\Gamma_1^*(\ER)+B_r(0)$. Here $K_1$ represents ${\rm mod}_1(K_1)$.
}
\label{fig:momenta_trunc}
\end{figure}
Additionally, there exists $r_m>0$ such that for $r+r_\Sigma>r_m$, $\Gamma_j^*(\ER)+B_r(0)$ becomes homotopically non-trivial \cite{MassattCarrLuskinOrtner2018}. Therefore, we should take $r<r_m-r_\Sigma$ to ensure a meaningful numerical scheme. 
It is also important to mention that if $q$ is near $K$ ($K'$), then $ \Omega_r^*(q)$ only corresponds to the $K$\;($K'$) valley, if ${\rm mod}_j(q)\notin\Gamma_j^*(\ER)+B_r(0)$, then $ \Omega_r^*(q)$ is empty. We then define the inclusion map
\begin{equation}
\label{inclusion}
    J_{V\gets U}:\ell^2(U)\to\ell^2(V),\qquad U\subset V\subset\Omega^*,
\end{equation}
and denote the {inclusion of $\Omega_r^*(q)$ into $\Omega^*$ by $J_r(q):=J_{\Omega^*\gets\Omega_r^*(q)}$}. The corresponding finite matrix is $J_r^*(q) \eH(q)J_r(q)$.

We further exploit the exponential decay of the interlayer hopping functions to sparsify the Hamiltonian matrix. Let $\tau\in\Z_+$ be a truncation parameter for hopping distance. 
We define the truncated hopping index set for $\tK\in\{K,K'\}$ valley by
\[
\B_\tau  := \Big\{n\in\Z^2: |\tK_1+2\pi A_1^{-T}n| \text{ is among the $\tau$ smallest magnitudes of } \tK_1+G_1\Big\}.
\]
The size of the index set, $\#\mathcal{B}_\tau$, is $O(\tau^2)$. Interestingly, numerical computations suggest that $\#\mathcal{B}_\tau$ is close to $O(\tau)$, but we are not aware of any proof of this. Additionally, note that $\B_\tau$ for $K'$ valley is the negative of $\B_\tau$ for $K$ valley, since $K_j' = -K_j$.
The restricted interlayer part is then given by
\[
[\eH^{(\tau)}_{12}(q)]_{G,G'} := [\eH_{12}(q)]_{G,G'}\delta_{\I(G)+\I(G')\in \B_\tau},
\]
where $\I:\R_1^*\cup\R_2^*\to\Z^2$ is the index map 
\[\I(G_j):=A_j^TG_j/2\pi,\qquad \G_j\in\R_j^*.\]

Finally, we construct the truncated matrix by
\begin{equation}
\label{finiteH}
    \trcH(q) := J^*_r(q)\begin{pmatrix} \eH_{11}(q)& \eH_{12}^{(\tau)}(q) \\ \Big({{\eH_{12}}^{(\tau)}}(q)\Big)^\dagger &\eH_{22}(q) \end{pmatrix}J_r(q).
\end{equation}
The DoS of $H$ at $E\in B_\Sigma$ can be well approximated by the DoS of $\trcH$, when there is a large moir\'{e} pattern, which occurs at small twist angle $\theta$. We require the following definition for moir\'e-periodic systems.
\begin{definition}
    The moir\'{e} reciprocal lattice and moir\'{e} Brillouin zone are
    \[
    \mR:= 2\pi\mm\Z^2,\qquad \mBZ:=2\pi\mm[0,1)^2
    \]
    with lattice matrix
    $$\mm:=A_2^{-T}-A_1^{-T}.$$
   The moir\'{e} Dirac points are
    \[
    K_{\M}:=K_1,\qquad  K'_{\M}:=K_2.
    \]
\end{definition}


We also define the maps $\Gmj : \mR\to\R_j^*$ and $\Gjm : \R_1^*\cup\R_2^*\to\mR$ as transformations between moir\'{e} reciprocal lattices and monolayer reciprocal lattices
\begin{equation}
\label{jM_trans}
\begin{split}
     &\Gmj(\mG):=A_j^{-T}\mm^{-1}\mG, \qquad\mG\in\mR,\\
    &\Gjm(G_j):=(-1)^j\mm A_j^TG_j,\qquad G_j\in\R_j^*.
\end{split}
\end{equation}
Using {periodicity of $\trcH(\cdot)$ with respect to moir\'{e} reciprocal lattices up to unitary transformation}, the starting momenta in \eqref{DOS:ms} {can be restricted into the much smaller space $\mBZ$} \cite{Massatt2023}. We restate the result for reference:
\begin{theorem}
\label{thm:truncation:dos}
Assume Assumption \ref{as:incommensurate} and assume TBG has hopping functions satisfying Assumptions \ref{assumption:interhop} and \ref{assumption:intrahop}. We define DoS for the truncated Hamiltonian at valley $\tK \in \{K,K'\}$ by
\begin{equation}
\label{DOS:finiteH}
    D_{\varepsilon,r}^{(\tau)}(E;\tK):=\nu^*\int_{\tilde K+\mBZ}{\rm Tr}\;\gauss(E-\trcH(q))\; dq.
\end{equation}
Consider $E\in B_\Sigma$, $\theta \ll 1$, and $\varepsilon\ll 1$.
Let $\htr\in\RR_+$ be the hopping truncation distance depending on $\tau\in\Z_+$.
Then there are constants $\gamma_h$, $\gamma_m$, and $\gamma_g$ corresponding to hopping truncation error, momentum truncation error, and Gaussian decay rates respectively such that 
\begin{equation} 
\label{eq:D_est}
    \Big|D_\varepsilon(E)-\sum_{\tK\in\{K,K'\}}D_{\varepsilon,r}^{(\tau)}(E;\tK)\Big|\lesssim  \varepsilon^{-2}e^{-\gamma_h\htr}+\varepsilon^{-4}e^{{-\gamma_m}r}+\varepsilon^{-1}e^{-\gamma_g\varepsilon^{-2}},
\end{equation}
where $D_\varepsilon(E)$ is defined in \eqref{eq:full_DoS} (equivalently \eqref{DOS:ms}). Further, $\gamma_m$ scales as $\theta^{-1}$.
\end{theorem}

Note that the trace in \eqref{DOS:finiteH} converges trivially since $\widehat{H}_r^{(\tau)}(q)$ is a finite-dimensional matrix. Estimate \eqref{eq:D_est} shows that the error on the right-hand side is small when $d_\tau$ and $r$ are large and $\varepsilon$ is small, as long as $\varepsilon$ does not $\downarrow 0$ too fast. Specifically, as long as
\begin{equation}
    \varepsilon \gg e^{- {\frac{1}{4}\gamma_m r}}, \qquad \varepsilon \gg e^{- {\frac{1}{2} \gamma_h d_\tau}},
\end{equation}
then the error in \eqref{eq:D_est} will $\rightarrow 0$ as $\varepsilon \downarrow 0$ and $r, d_\tau \rightarrow \infty$. It is important to note that increasing $r$ beyond a certain value will result in an infinite-dimensional ``truncated'' matrix, which cannot be diagonalized exactly. However, in practice, the error in \eqref{eq:D_est} becomes very small for small $\varepsilon$ well before this point.


\section{Derivation of continuum models from truncated reciprocal space models}

In this section, we will propose an approximation scheme to obtain the higher order continuum models. As seen in Figure \ref{fig:workflow}, we have obtained the truncated Hamiltonian $\trcH$ by the Bloch transform, an unfolding procedure, and the reciprocal space truncation in Section 2. The next step is to perform the Taylor expansion at the Dirac points of the incommensurate Brillouin zone to obtain a polynomial approximation Hamiltonian $\tpH$. Subsequently, we will use a similar inverse process to derive the continuum model, including an infinite extension to define the Hamiltonian on $\RR^2$, a folding process back to the momentum space, and the inverse Fourier transform to obtain the PDE form.

\begin{figure}[htb!]
\centering
\includegraphics[height=2.8cm]{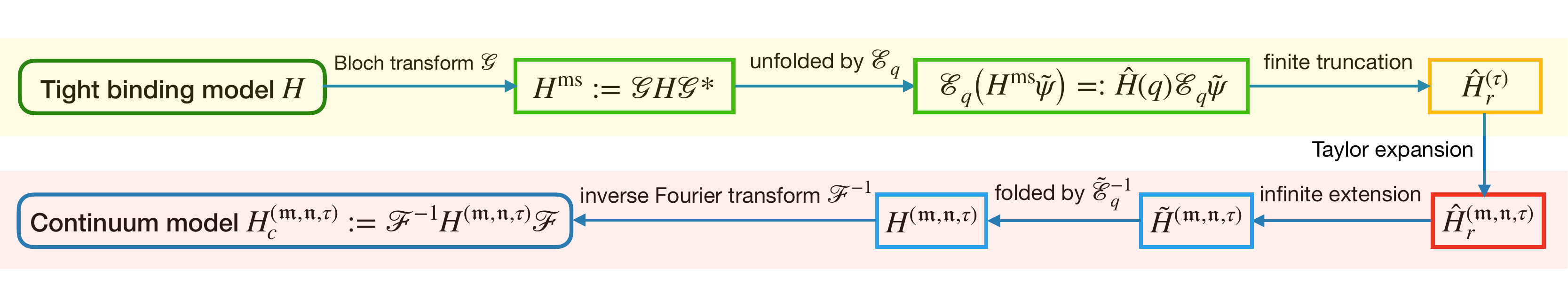}  
\setlength{\abovecaptionskip}{0pt} 
\caption{
Procedure for the construction of the continuum model.
}
\label{fig:workflow}
\end{figure}

\subsection{Taylor Expansion}


We begin with a momenta polynomial approximation of \eqref{finiteH}.
We note that the DoS \eqref{DOS:finiteH} is an integral over the moir\'{e} Brillouin zone $\mBZ$. We hence can only consider the behavior of the Hamiltonian in moir\'{e} reciprocal space. We rewrite \eqref{finiteH} by the moir\'{e} reciprocal lattices and do the Taylor expansion around $\tK\in\{K,K'\}$ valley to obtain the polynomial approximation model.

Before expanding the Hamiltonian, we introduce the 2-dimensional multi-index notations for $\beta\in\N^2$, $\xi\in\RR^2$ and $h:\RR^2\to\C, h\in C(\RR^2)$,
\begin{align*}
    |\beta|:=\beta_1 + \beta_2,\qquad \beta!:=\beta_1 !\beta_2!,\qquad\xi^\beta := \xi_1^{\beta_1}\xi_2^{\beta_2}, \qquad 
    D^\beta h := \frac{\partial^{|\beta|}h}{\partial \xi_1^{\beta_1}\partial \xi_2^{\beta_2}}.
\end{align*}
Let $\iip\in\N$ and $\ijp\in\N$ be the intralayer and interlayer expansion orders, respectively.
For $G\sigma,G'\sigma' \in \Omega_j^*$, we 
do the Taylor expansion for $\tilde{h}^{jj}$ with respect to $q+\Gjm(G)$ around $\tK_j$, 
\begin{align*}
   \bthop(q+G) &=  e^{-{i(G-\Gjm(G))}\cdot(\tau_{\sigma}-\tau_{\sigma'})}\bthop(q+\Gjm(G))\\
    &\approx e^{-{i(G-\Gjm(G))}\cdot(\tau_{\sigma}-\tau_{\sigma'})} P_{j,\sigma\sigma'}^\iip\big(q-\tK_j+\Gjm(G)\big),
\end{align*}
where we denote the $\iip$-order Taylor polynomial approximation of $\bthop$ by 
\begin{equation}
    \label{intra_poly}
    P_{j,\sigma\sigma'}^\iip(\xi) :=  \sum_{|\beta|\leq \iip}\frac{D^\beta\bthop(\tK_j)}{\beta !}\xi^\beta.
\end{equation}
For $G\sigma \in \Omega_1^*$, $G'\sigma' \in \Omega_2^*$, let
\[
{\GG_1} := G+G'-\Gjm(G)\in\R_1^*,
\]
we do the Taylor expansion for $\hat{h}^{12}$ with respect to $q+\Gjm(G)+\GG_1$ around $\tK_1+\GG_1$,
\begin{align*}
   c_1 c_2 {\interh}(q+G+G')\delta_{\I(G)+\I(G')\in \B_\tau}&=  c_1 c_2 {\interh}(q+\Gjm(G)+\GG_1)\delta_{\I(\GG_1)\in \B_\tau}\\
   &\approx U^{\ijp}_{\GG_1,\sigma\sigma'}\big(q-\tK_1+\Gjm(G)\big)\delta_{\I({\GG_1})\in \B_\tau}
\end{align*}
where we denote the $\ijp$-order Taylor polynomial approximation of $c_1c_2\interh$ by
\begin{equation}
    \label{inter_poly}
    U^{\ijp}_{\GG_1,\sigma\sigma'}(\xi):=c_1 c_2 \sum_{|\beta|\leq \ijp}\frac{D^\beta {\interh}(\tK_1+\GG_1)}{\beta !}\xi^\beta.
\end{equation}
Hence, the polynomial approximation Hamiltonian is given by
\begin{equation}
\label{GBM}
  \tpH(q) := J^*_r(q)\begin{pmatrix} \eH^{(\iip)}_{11}(q)& \eH_{12}^{(\ijp,\tau)}(q) \\ \Big({ \eH_{12}^{(\ijp,\tau)}}(q)\Big)^\dagger &\eH^{(\iip)}_{22}(q) \end{pmatrix}J_r(q),
\end{equation}
where
\begin{align*}
    [\eH^{(\iip)}_{jj}(q)]_{G\sigma,G'\sigma'}:=e^{-{i(G-\Gjm(G))}\cdot(\tau_{\sigma}-\tau_{\sigma'})}P_{j,\sigma\sigma'}^\iip\big(q-\tK_j+\Gjm(G)\big)\delta_{GG'},
\end{align*}
and 
\[
[\eH_{12}^{(\ijp,\tau)}(q)]_{G\sigma,G'\sigma'} :=  \HT_{G,G'}^{\sigma\sigma'}U^{\ijp}_{\GG_1,\sigma\sigma'}\big(q-\tK_1+\Gjm(G)\big)\delta_{\I({\GG_1})\in \B_\tau}.
\]

We note that for $q\in \tK + \mBZ$ and $G\sigma\in\Omega_r^*(q)$, we have $|q-\tK_j+\Gjm(G)|\leq r_\Sigma + r$, $G\in\R^*_{{\rm F}_j}$. Since the Taylor expansion error scales as $O\big(|q -\tK_j+ \Gjm(G)|^{\mathfrak{j}+1}\big), \mathfrak{j} = \iip,\ijp$, the expansion accuracy will also be determined by the reciprocal space truncation. The following lemma gives the precise error analysis of the polynomial approximation.
\begin{lemma}
\label{lemma:convergence:eigenvalues}
Assume Assumption \ref{as:incommensurate} and assume TBG has hopping functions satisfying Assumption \ref{assumption:interhop} and \ref{assumption:intrahop}. For $q\in \tilde{K}+\mBZ$, $\tilde{K}\in \{K,K'\}$, $\theta \ll 1$, and $r<\min\{\gamma_1,\gamma_2,\gamma_{12}, r_m\}-r_\Sigma$, there is a constant $C$ independent of $\iip,\ijp,\tau$,
such that the error of the polynomial approximation is 
\begin{equation}
    \label{error:eig}
    \|\trcH(q)-\tpH(q)\|_{\op}\leq C\Big((\iip+1) \big(\max_{j={1,2}}{\gamma_j^{-1}}(r_\Sigma + r)\big)^{\iip+1} + \# \B_\tau(\ijp+1)\big(\gamma_{12}^{-1} (r_\Sigma + r)\big)^{\ijp+1}\Big).
\end{equation}
\end{lemma}
\begin{proof}
    Proof is in Appendix \ref{proof:error:eigs}.
\end{proof}
Estimate \eqref{error:eig} in particular implies convergence of the DoS of the polynomially approximated Hamiltonian $\tpH$ to that of $\trcH$. Combining this result with Theorem \ref{thm:truncation:dos} will allow us to relate the DoS of $\tpH$ with the full DoS \eqref{eq:full_DoS} in our main result Theorem \ref{thm:convergence:dos}.

\subsection{Continuum Model}
We now derive the continuum model from \eqref{GBM}. It is not immediately obvious how to do this, as $\tpH$ is defined on $\Gamma_j^*(\ER)+B_r(0)$ (see the momentum domain in the polynomial expansion), while the continuum model is defined on $\mathbb{R}^2$. However, we find that the density of states of these models can be related as long as we focus on low energies, i.e., energies close to the monolayer Dirac energy. To make this precise, we have to extend the Hamiltonian to all of $\mathbb{R}^2$ in such a way that it remains unchanged for small momenta (i.e., momenta close to the monolayer Dirac point) but becomes invertible for large momenta. We will then apply the inverse Fourier transform to obtain the continuum model.


For $\tK\in\{K,K'\}$, let $\ex(\xi) := \xi\chi_r(\xi)$ with $\chi_r(\xi)$ smooth and 1 on 
$$\bigcup_{j=1,2} \bigg\{ \Big\{\Gamma_j^*(\ER)+B_r(0)+\tK_j-{\rm mod}_j(\tK_j)\Big\}\cap B_{r_\Sigma+r}(\tK_j)\bigg\},$$
and compactly supported on 
$$\bigcup_{j=1,2} \bigg\{\Big\{\Gamma_j^*(\ER)+B_{r'}(0)+\tK_j-{\rm mod}_j(\tK_j)\Big\}\cap B_{r_\Sigma+r'}(\tK_j)\bigg\},$$ with $0<r'-r\ll
1$ so that the momenta extended by $r'$ do not contribute to the DoS. Here $\tK_j-{\rm mod}_j(\tK_j)$ shifts the momenta in the Brillouin zone to the reciprocal lattice unit cell containing $\tK_j$, and the intersection ensures that the momenta correspond to $\tK$ (see Figure \ref{fig:momenta_trunc}). We modify the expansion polynomials by
\begin{align}
\label{polymod}
\cP^{\iip}_j(\xi):= P_j^1(\xi) + (P_j^\iip-P_j^1)(\ex(\xi)),\qquad
\cU^{\ijp}_{\GG_1}(\xi):=
 U^{\ijp}_{\GG_1}(\ex(\xi)),
\end{align}
and denote the corresponding Hamiltonian parts by $\widetilde{H}^{(\iip)}_{jj}$ and ${ \widetilde{H}_{12}^{(\ijp,\tau)}}$. Then the infinite matrix is 
\[
\bmH(q):=\begin{pmatrix} \widetilde{H}^{(\iip)}_{11}(q)& \widetilde{H}_{12}^{(\ijp,\tau)}(q) \\ \Big({ \widetilde{H}_{12}^{(\ijp,\tau)}}(q)\Big)^\dagger &\widetilde{H}^{(\iip)}_{22}(q) \end{pmatrix}.
\]
The reason for modifying the expansion polynomials as in \eqref{polymod} is to ensure that $\tilde{P}_j^\mathfrak{m}$ is invertible for large $|\xi|$, which will be important in the proofs below.

We notice that $\bmH$ could be rewritten by the moir\'{e} reciprocal lattice, i.e., for $\mG\sigma,\mG'\sigma'\in{\Omega^*_{j{\rm M}}}:=\mR\times\A_j$,
\begin{equation}
    \label{intra:GBM}
     [\widetilde{H}^{(\iip)}_{jj}(q)]_{\mG\sigma,\mG'\sigma'}=e^{i\Gmj({\mG})\cdot(\tau_{\sigma'}-\tau_{\sigma})}\cP_{j,\sigma\sigma'}^\iip\big(q-\tK_j+(-1)^{{\rm F}_j}\mG\big)\delta_{\mG\mG'},
\end{equation}
and for $\mG\sigma\in\Omega^*_{1{\rm M}}, \mG'\sigma'\in\Omega^*_{2{\rm M}}$,
\begin{equation}
   \label{inter:GBM}
     [\widetilde{H}^{(\ijp,\tau)}_{12}(q)]_{\mG\sigma,\mG'\sigma'}=\HT_{\mathfrak{G}_2(\mG),\mathfrak{G}_1(\mG')}^{\sigma\sigma'}
     \cU^{\ijp}_{\mathfrak{G}_1(\GG_{\rm M}),\sigma\sigma'}\big(q-\tK_1+\mG\big)\delta_{\I({\GG_{\rm M}})\in \B_\tau}, 
\end{equation}
where $\GG_{\rm{M}}:= \mG+\mG'$ and $\I(\mG) := (2\pi\mm)^{-1}\mG$ for $\mG\in\mR$. We then define a map $\mE_{q,j}:\MSspace_j\to\ell^\infty({\Omega^*_{j{\rm M}}})$ analogous to $\E_q$ but sampled by moir\'{e} reciprocal lattices
\[
\{\mE _{q,j}\bw_j\}_{\mG\sigma}:=e^{-i\Gmj(\mG)\cdot\tau_{\sigma}}\bw_j^\sigma(q-\tK_j+(-1)^{{\rm F}_j}\mG),
\] 
and let
\[
\mE_q := \begin{pmatrix} \mE_{q,1}& 0 \\ 0 & \mE_{q,2}\end{pmatrix}.
\]
The following lemma gives the folded (momentum space) form of the approximation model. 

\begin{lemma}
\label{lemma:continuum model}
    We construct the operator $H^{(\iip,\ijp,\tau)} : \mathcal{F} H^1(\mathbb{R}^2;\mathbb{C}^4) \rightarrow L^2(\mathbb{R}^2;\mathbb{C}^4)$, where $\mathcal{F}$ is the Fourier transform, decomposed sheet-wise as
    \begin{equation}
    \label{MS:HC}
  H^{(\iip,\ijp,\tau)} =\begin{pmatrix} H_{11}^{(\iip)}& H_{12}^{(\ijp,\tau)} \\ \Big(H_{12}^{(\ijp,\tau)}\Big)^\dagger &H_{22}^{(\iip)} \end{pmatrix}
\end{equation}
to be the Hamiltonian such that for $\bw \in \oplus_{j=1}^2 \MSspace_j $,
\[
\mE_q \big(H^{(\iip,\ijp,\tau)}\bw\big)=\bmH(q)\mE_q\bw.
\]
The operator $H^{(\iip,\ijp,\tau)}$ is composed of translation operators and multiplication operators given by
\begin{align}
    &H_{jj}^{(\iip)}=\cP^{\iip}_j(\cdot),\qquad j = 1,2,\\[1ex]
    &H_{12}^{(\ijp,\tau)} = \sum\limits_{\I(\mG)\in\B_\tau}\HT_{\mG} \odot \cU^{\ijp}_{\mathfrak{G}_1(\mG)}(\cdot)T_{s_{\mG}},
\end{align}
where $\HT_{\mG}^{\sigma\sigma'} := e^{i\mathfrak{G}_1(\mG)\cdot\tau_{\sigma}}e^{-i\mathfrak{G}_2({\mG})\cdot\tau_{\sigma'}}$ for $\sigma\in\A_1,\sigma'\in\A_2$ and $s_{\mG}:=\tK_1-\tK_2-\mG$.

\end{lemma}
\begin{proof}
    See Appendix \ref{compt:MS:GBM}.
\end{proof}

Finally, let $D:=-i(\partial_{x_1},\partial_{x_2})$ and $\ex^n(D)f(x):=\frac{1}{(2\pi)^2}\int_{\RR^2}e^{i\xi x}\ex^n(\xi)\hat{f}(\xi)\;d\xi$.
By the inverse Fourier transform, we obtain the continuum model for the $\tK$ valley
\begin{align}
    \label{real:GBM}
  &H^{(\iip,\ijp,\tau)}_c:= \F^{-1}H^{(\iip,\ijp,\tau)}\F\\[1ex]
  \nonumber
  &=
  \begin{pmatrix} \cP^{\iip}_1\big(D\big)& \sum\limits_{\I(\mG)\in\B_\tau}\HT_{\mG} \odot \cU^{\ijp}_{\mathfrak{G}_1(\mG)}\big(D-{s_{\mG}}\big)e^{-i{{s_{\mG}}}\cdot x}\\ \sum\limits_{\I(\mG)\in\B_\tau}\HT^\dagger_{\mG} \odot\overline{\cU^{\ijp}_{\mathfrak{G}_1(\mG)}\big(D-{s_{\mG}}\big)}e^{i{{s_{\mG}}}\cdot x} &\cP^{\iip}_2\big(D\big) \end{pmatrix}.
\end{align}
We refer to this as the $(\iip, \ijp, \tau)$ continuum model.
We mention that the $r$ dependence is omitted in the continuum model Hamiltonian. This is because $r$ is mainly used to improve numerical efficiency, and is always chosen to ensure the momentum truncation converges. Therefore, $r$ does not fundamentally affect the accuracy of the continuum model.  

In the following example, we can see that the BM model is a leading order continuum model. Therefore, we also refer to the $(1,0,1)$ continuum model as the BM model.
\begin{example}
    We consider the simplified TBG model in \cite{watson2023}. This simplified model has three nearest hoppings for the intralayer, and with a radial decay function for the interlayer such as \eqref{toy_hs_inter}. Specifically, the reciprocal space Hamiltonian for this model is
    \begin{align}
    \label{toy_rs}
    &[\eH_{jj}]_{G,G'}  = -t\begin{pmatrix} 0 & F_j(\cdot+G) \\ \overline{F_j(\cdot+G)}& 0 \end{pmatrix}\delta_{GG'}, \quad F_j(q) := e^{iq\cdot(\tau_{jB}-\tau_{jA})}(1+e^{-iq\cdot a_{j,1}}+e^{-iq\cdot a_{j,2}}),\\[1ex]\nonumber
    &[\eH_{12}]_{G,G'} =c_1c_2 {\hat{h}_{12}}(\cdot+G+G')\HT_{G,G'}.
    \end{align}

    We first derive the $(1,0,1)$ continuum model for the $K$ valley. We can easily obtain that
\[
P_j^1(q)= v\sigma_{\theta_j}\cdot q,\quad \sigma_{\theta_j}\cdot q:=\begin{pmatrix} 0& e^{i\theta_j}(q_1-iq_2) \\ e^{-i\theta_j}(q_1+iq_2) & 0\end{pmatrix},\quad v:=\frac{\sqrt{3}}{2}at,\quad \theta_j:=(-1)^j\frac{\theta}{2},
\]
and
\[
U^0_{\GG_1}(q)= w\begin{pmatrix} 1& 1 \\ 1 &1 \end{pmatrix},\qquad  w := c_1c_2\hat{h}_{12}(K_1+\GG_1), \qquad \I({\GG_1})\in \B_1,
\]
where $\B_1 = \{(0,0),\,(0,1),\,(-1,0)\}$.
Here we use $w$ as the interlayer coupling since the three nearest hopping terms have equal strength. 
Then we have 
\[
H^{(1,0,1)}_c =\begin{pmatrix} v\sigma_{-\theta/2}\cdot D& w\sum_{j=1}^3{\HT}_je^{-is_j\cdot x} \\ \overline{w}\sum_{j=1}^3{\HT}^\dagger_je^{is_j\cdot x} &v\sigma_{\theta/2}\cdot D \end{pmatrix},
\]
where the matrices are
\begin{align*}
    {\HT}_1:=\begin{pmatrix} 1 & 1 \\ 1 & 1 \end{pmatrix}, \qquad
    {\HT}_2:=\begin{pmatrix} 1 & e^{-i\phi} \\ e^{i\phi} & 1 \end{pmatrix},\qquad
    {\HT}_3:=\begin{pmatrix} 1 & e^{i\phi} \\ e^{-i\phi} & 1 \end{pmatrix},
\end{align*}
with $\phi=2\pi/3$, and the vectors are
\begin{equation*}
    s_1 := K_1-K_2,\qquad s_2:=K_1-K_2-2\pi\mm(0,1)^T,\qquad s_3:=K_1-K_2-2\pi\mm(-1,0)^T.
\end{equation*}
The three vectors could also be explicitly given by
\[
s_1 = |\Delta K|(0,-1)^{T},\qquad s_2=|\Delta K|\Bigg(\frac{\sqrt{3}}{2},\frac{1}{2}\Bigg)^{T},\qquad s_3=|\Delta K|\Bigg(-\frac{\sqrt{3}}{2},\frac{1}{2}\Bigg)^{T},
\]
where $|\Delta K|:=2|K|\sin(\theta/2)$ is the distance between the Dirac points of the layers. We can see that $H^{(1,0,1)}_c$ for the $K$ valley is highly similar to the BM model given in \cite{Bistritzer2011}.

Observe that $K'=-K$, we can immediately obtain the $(1,0,1)$ continuum model for the $K'$ valley
\[
H^{(1,0,1)}_c =\begin{pmatrix} v\sigma'_{-\theta/2}\cdot D& w\sum_{j=1}^3{\HT'_j}e^{-is'_j\cdot x} \\ \overline{w}\sum_{j=1}^3{\HT'_j}^\dagger e^{is_j'\cdot x} &v\sigma'_{\theta/2}\cdot D \end{pmatrix},
\]
where 
\[
(\sigma'_{\theta_j}\cdot q)_{\sigma\sigma'}:=-\overline{(\sigma_{\theta_j}\cdot q)_{\sigma\sigma'}},\qquad
 (\HT'_j)_{\sigma\sigma'} := \overline{(\HT_j)_{\sigma\sigma'}},\qquad s'_j:=-s_j.
\]
 \end{example}

\subsection{Convergence of the Density of States}
\label{section:convergence}
We finally discuss the convergence of the density of states of the continuum model. Similar with \eqref{DOS:finiteH}, that is approximated by
\begin{equation}
\label{DOS:GBM}
    D^{(\iip,\ijp,\tau)}_{\varepsilon}(E;\tilde{K}):=\nu^*\int_{\tilde{K}+\mBZ}{\rm Tr}\;\gauss(E-\bmH(q))\;dq.
\end{equation}
We can see that the error of $D^{(\iip,\ijp,\tau)}_{\varepsilon}$ comes from three parts: 
truncation error, polynomial approximation error, and smooth extension error. The first one has been given by Theorem \ref{thm:truncation:dos}, the second one can be estimated by Lemma \ref{lemma:convergence:eigenvalues}, and the last one can be analyzed using the ``ring decomposition technique" proposed by \cite{MassattCarrLuskinOrtner2018}. 
\begin{theorem}
    \label{thm:convergence:dos}
Assume Assumption \ref{as:incommensurate} and assume TBG has hopping functions satisfying Assumption \ref{assumption:interhop} and \ref{assumption:intrahop}. Consider $E\in B_\Sigma$, $\theta\ll 1$, and $\varepsilon\ll 1$. 
Let $\htr\in\RR_+$ be the hopping truncation distance depending on $\tau\in\Z_+$. Let $\iip,\ijp\in\N$ be the intralayer and interlayer expansion orders respectively. Then for $r<\min\{\gamma_1,\gamma_2,\gamma_{12},r_m\}-r_\Sigma$, there are constants $\gamma_h$, $\gamma_m$, and $\gamma_g$ corresponding to hopping truncation error, momenta truncation error, and Gaussian decay rates respectively, such that
\begin{align}
    \label{error:dos}
    |D_\varepsilon(E)-\sum_{\tK\in\{K,K'\}}D^{(\iip,\ijp,\tau)}_{\varepsilon}(E;\tK)|\leq  C_{\Sigma,\theta,\tau,r}\left(\varepsilon^{-2}e^{-\gamma_h\htr}+\varepsilon^{-4}e^{{-\gamma_m}r}+\varepsilon^{-1} e^{-\gamma_g\varepsilon^{-2}}\right)
    \\\nonumber
    +\varepsilon^{-2}C_{\Sigma,\theta}\left((\iip+1) \left(\max_{j={1,2}}{\gamma_j^{-1}}(r_\Sigma + r)\right)^{\iip+1} + \#\B_\tau(\ijp+1)\left(\gamma_{12}^{-1} (r_\Sigma+r)\right)^{\ijp+1}\right),
\end{align}
where $C_{\Sigma,\theta,\tau,r}$ is a constant dependent of $\Sigma,\theta,\tau,r$, and $C_{\Sigma,\theta}$ is a constant dependent of $\Sigma,\theta$
\begin{equation}
    C_{\Sigma,\theta} := C\max_{
    \tK\in\{K,K'\}}\#\Omega_{0}^*(\tilde{K}).
\end{equation}
\end{theorem}
\begin{proof}
    Proof is in Appendix \ref{proof:error:dos}.
\end{proof}
We note that the constant for the polynomial approximation error $C_{\Sigma,\theta}$ scales as $O(\Sigma\theta^{-2})$. Although we consider $\theta\ll1$, since we only focus on a very small energy region around Dirac points, $C_{\Sigma,\theta}$ is still a small constant. 
We note as long as the hopping functions do not have asymptotic behavior with $\theta$ as is the case for the simplified and Wannier tight-binding models we consider, then $\gamma_m \sim \theta^{-1}$. Hence we can have $r$ proportional to $\theta$, and $r_\Sigma$ is small as interlayer tunneling is weak in Van der Waals systems. Consequently, small Taylor expansion orders $\iip,\ijp$ are sufficient to ensure convergence. We mention that the inclusion of mechanical relaxation effects alters this asymptotic \cite{Massatt2023}. 

\medskip

\section{Numerical Simulations}

In this section, we shall give the high accuracy continuum models for the simplified TBG and the Wannierized TBG. We will demonstrate the conclusions by convergence results and band structures at the $K$ valley for the magic angle ($\theta=1.1^\circ$). The results are similar for the $K'$ valley.
For the simplified TBG, we consider the model given by \eqref{toy_rs} and use the tight-binding parameters given in \cite{Bistritzer2011} and \cite{watson2023}. For the Wannierized TBG, we consider the monolayer tight-binding hopping parameters up to the fourth nearest neighbor, with values obtained from a previous first-principles study of graphene \cite{Fang2016}, and interlayer hopping coupling is also derived from the first-principles calculations \cite{Fang2016}. As seen in Figure \ref{fig:inter_coup}, the interlayer hopping couplings of these two TBG models have different decay and symmetry. Therefore, we can anticipate that their BM models will exhibit different accuracy on the band structure near the flat band, and their high-order corrections are also different.

\begin{figure}[htb!]
\centering
\includegraphics[height=4cm]{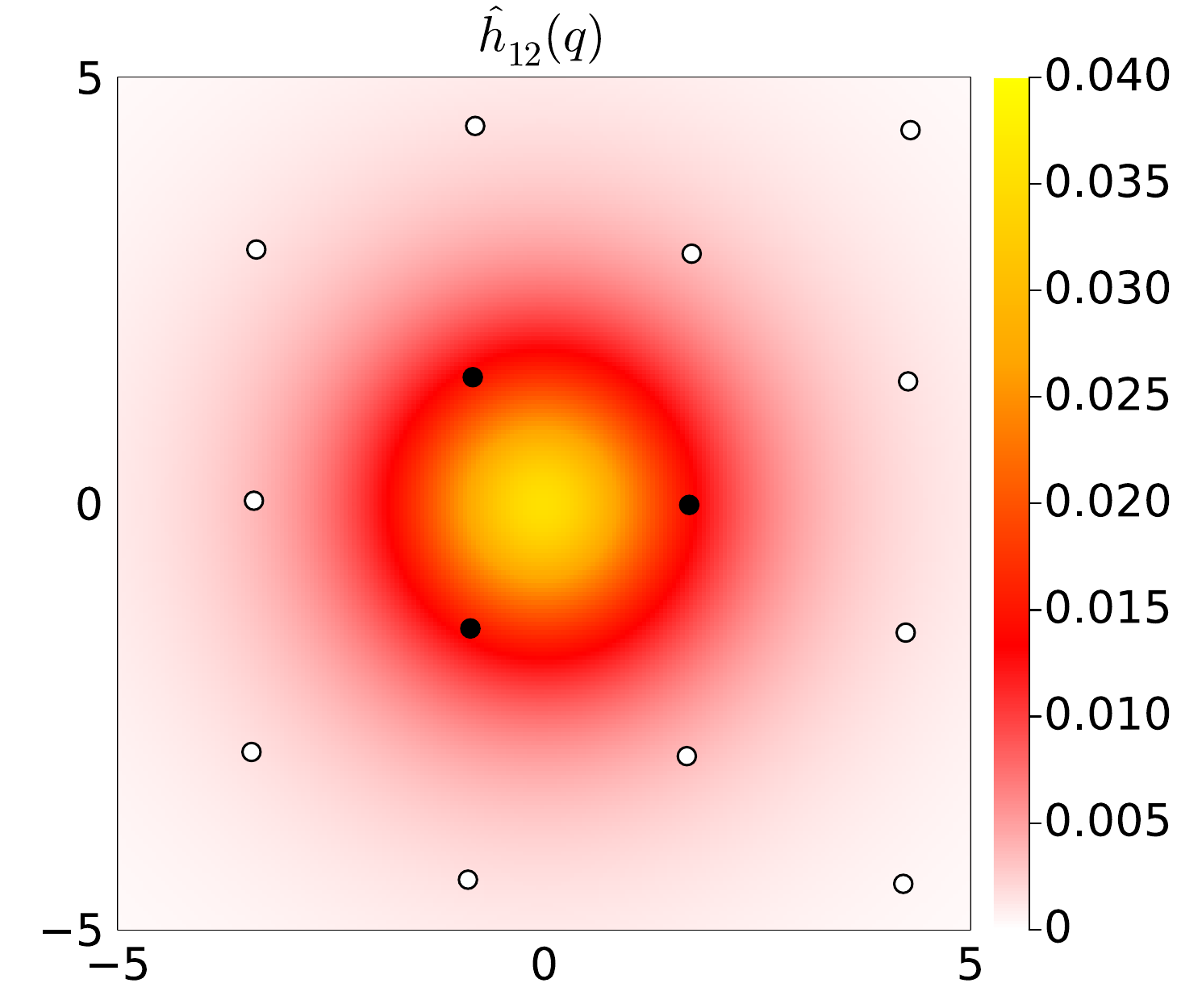} 
\includegraphics[height=4.05cm]{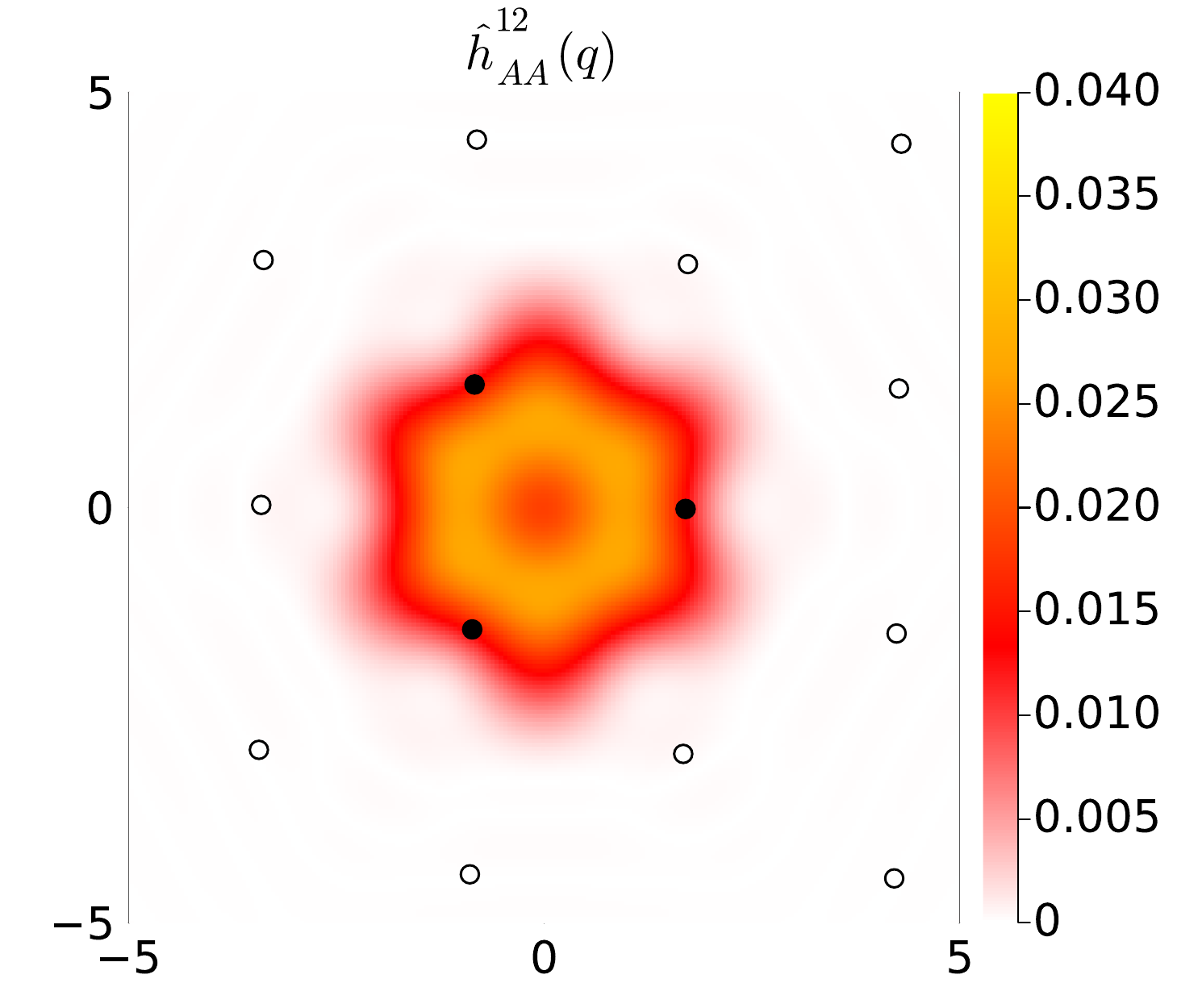}  
\includegraphics[height=4.05cm]{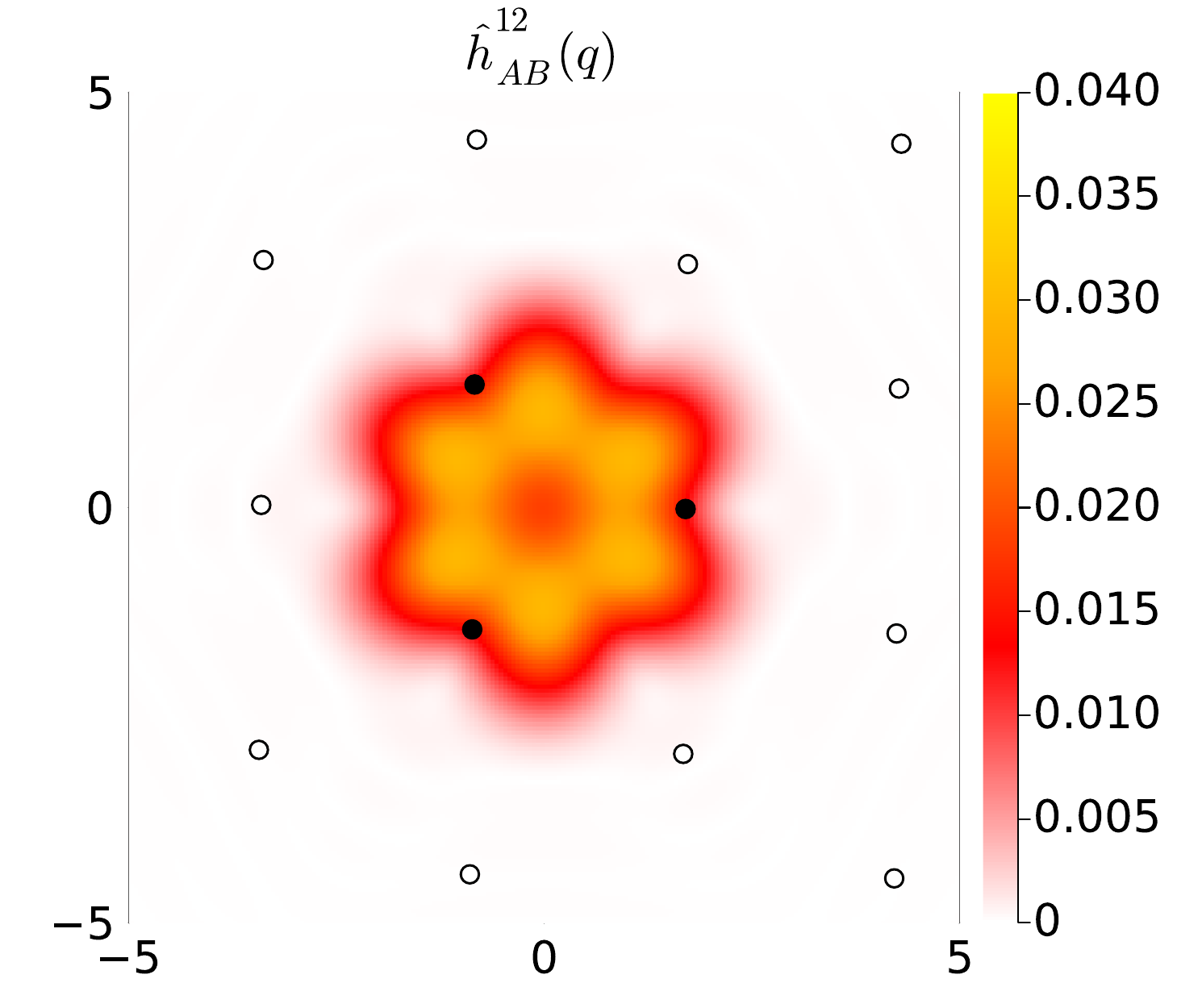}  
\setlength{\abovecaptionskip}{0pt} 
\caption{
Magnitude of interlayer coupling in momentum space for the simplified TBG, AA and AB parts of the Wannierized TBG, respectively. The balls represent all $K_1+G_1$, where the black balls are hoppings corresponding to $\B_1$.
}
\label{fig:inter_coup}
\end{figure}

In particular, we will use a more succinct method to truncate the reciprocal space for TBG. Let $\Lambda>0$ be the new momentum truncation parameter, then the basis set could be described by a fixed truncation radius
\begin{align*}
    \Omega^*_\Lambda:=\{G\sigma\in\Omega^*:|\Gjm(G)|<\Lambda\},
\end{align*}
and the corresponding inclusion is $J_\Lambda:=J_{\Omega^*\gets\Omega_\Lambda^*}$.
The reason for this is that the conical bands of graphene lead to a linear relationship between momentum cutoffs and energy range. We denote the high-symmetry lines of the moir\'{e} Brillouin zone at a single moir\'{e} $K_{\rm M}$ valley by $\LS$ ($K_{\rm M}\to\Gamma_{\rm M}\to M_{\rm M}\to K_{\rm M})$, and define the relative error for momentum truncation along $\LS$ by
\begin{equation*}
    {\rm Err}(\Lambda,\Sigma):=\max_{
    \substack{q\in \LS\\[0.5ex]j:|\epsilon_j(q)|\leq\Sigma}}\big|\epsilon_j(q)-\epsilon_j(\Lambda,q)\big|,
\end{equation*}
where $\epsilon_j(q)$ and $\epsilon_j(\Lambda,q)$ are the $j$-th eigenvalues of $\widehat{H}(q)$ and $J^*_\Lambda\widehat{H}(q)J_\Lambda$, respectively. In Figure \ref{fig:rcut_test}, we show the relative error as a bi-variate function of $\Lambda$ and $\Sigma$ for the Wannierized TBG. We can observe that $\Lambda$ grows linearly with $\Sigma$ when converging, which validates the rational choice of $\Omega^*_\Lambda$. We also note that the error does not decrease continuously with $\Lambda$, this is because the cardinality of $\Omega^*_\Lambda$ increases piecewisely with $\Lambda$.
\begin{figure}[htb!]
\centering
\includegraphics[height=5.2cm]{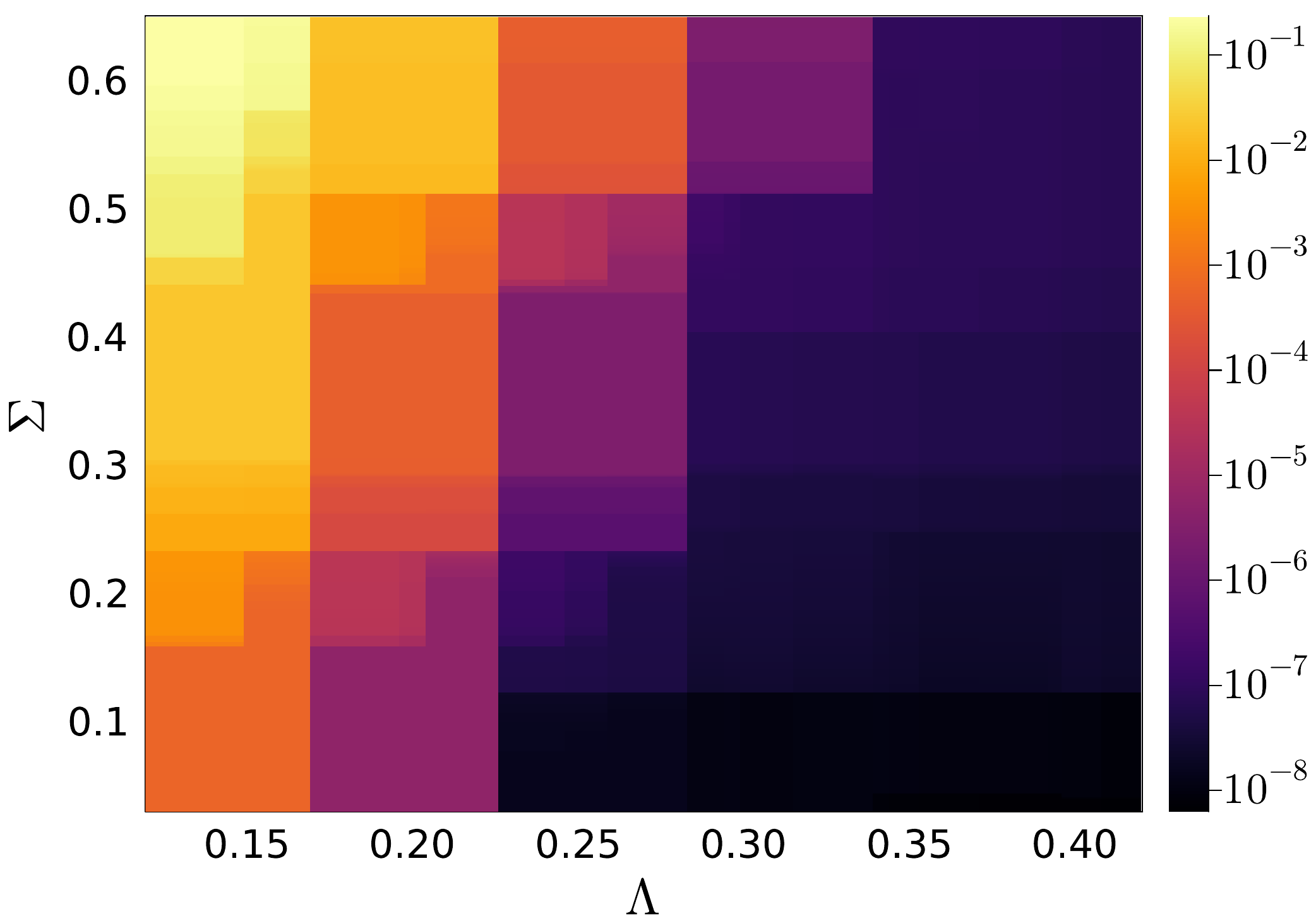 } 
\setlength{\abovecaptionskip}{0pt} 
\caption{
Relative error ${\rm Err}(\Lambda,\Sigma)$ for the Wannierized TBG at $\theta=1.1^\circ$.
}
\label{fig:rcut_test}
\end{figure}
In the following, all the simulations will use $\Omega_\Lambda^*$ as the momentum truncation basis set. Since we are only concerned with the parameters selection for the continuum model, we will omit the subscript of momentum truncation $r/\Lambda$ (chosen to ensure the momentum truncation error converges) to simplify the notation.

To select the ``optimal" parameters for high accuracy models, we need to know the contribution of each parameter to convergence. Therefore, we first use the momentum space model $\eH$ as a reference and compare it with $\eH^{(\tau)}$ to evaluate the hopping truncation error. Next, we fix the hopping truncation parameter $\tau=\tau_0$, and compare $\eH^{(\tau_0)}$ with $\eH^{(\iip,\infty,\tau_0)}$ ($\eH^{(\infty,\ijp,\tau_0)}$) to evaluate the intralayer (interlayer) expansion error. Here, $``\infty"$ indicates that there is no polynomial approximation for the corresponding part.
Figure \ref{fig:error} shows the relative errors for 6 electron eigenvalues closest to the Fermi energy along $\LS$ for three models. We observe that the errors decrease exponentially with $\tau$, $\iip$, $\ijp$, which is consistent with the convergence analysis. Although the model becomes more accurate with larger parameters, the model complexity is also higher. From the pictures, we can see that there are some ``optimal" parameters to balance the trade-off. For example, we can take $\iip=2,\ijp=1,\tau=6$ for the simplified TBG and $\iip=2,\ijp=1,\tau=2$ for the Wannierized TBG, to achieve the desired accuracy without introducing too much complexity.

\begin{figure}[htb!]
\centering
\begin{tabular}{cc}
\includegraphics[width=4.8cm]{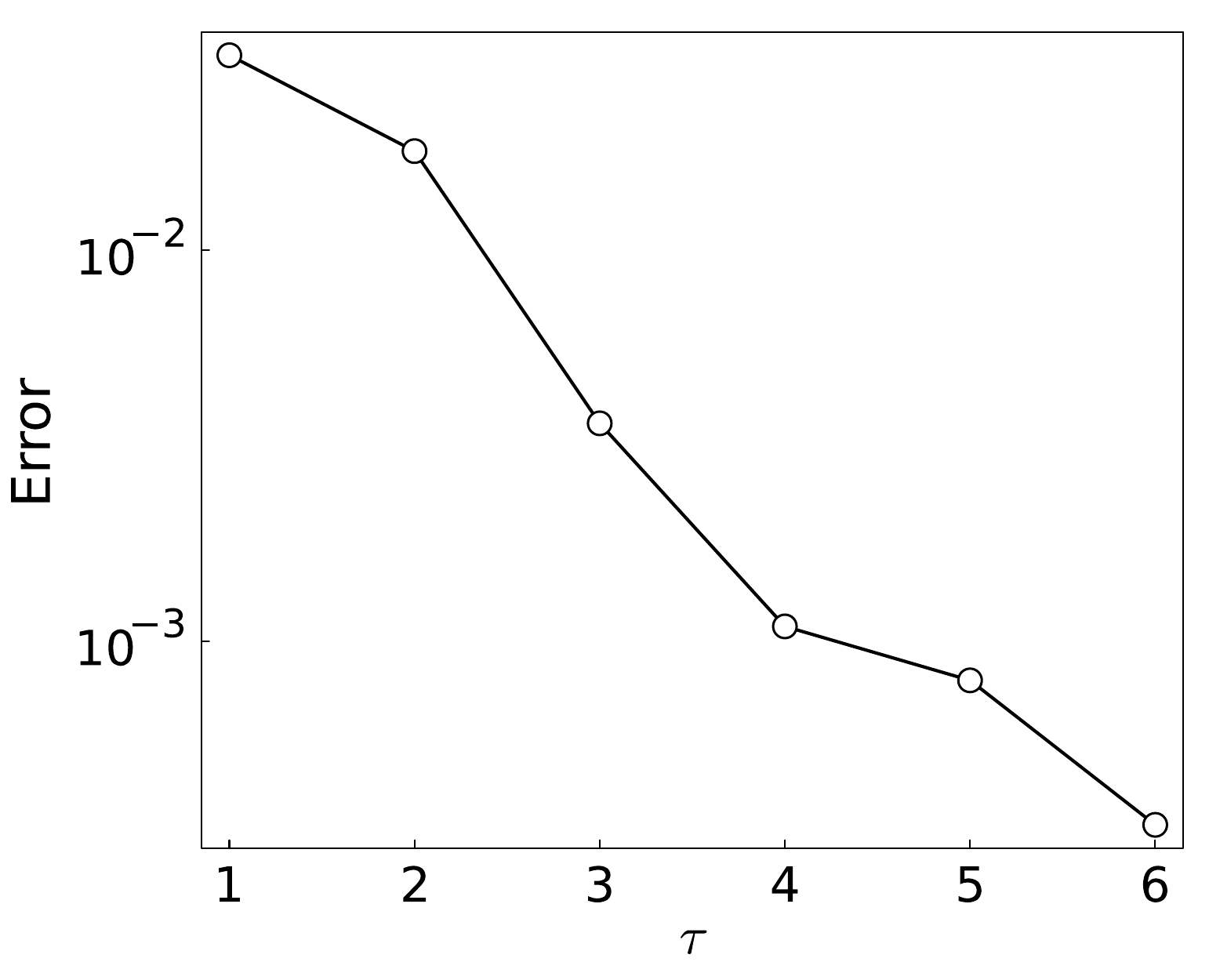}  
\hskip 1cm & 
\includegraphics[width=4.8cm]{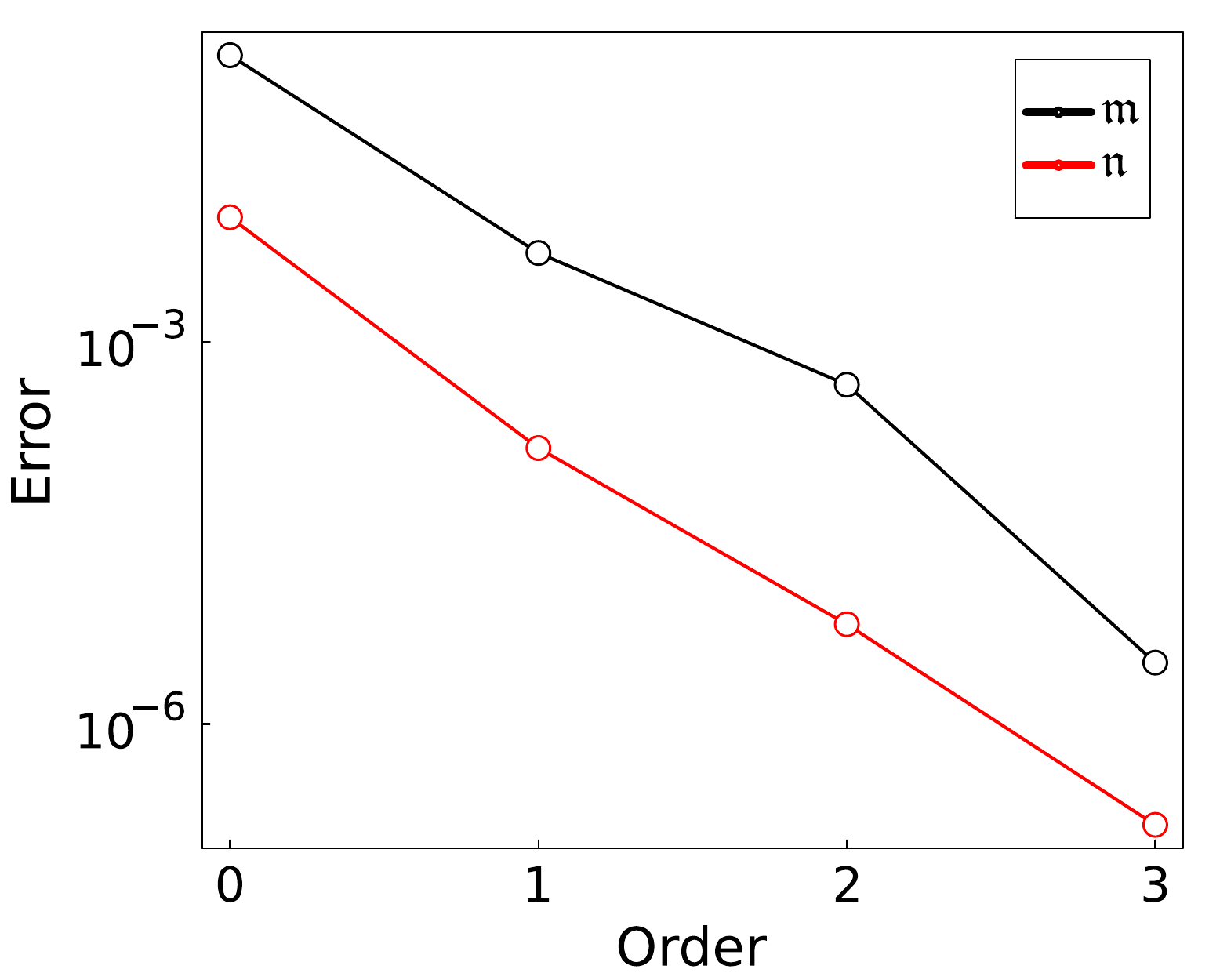} \\
\includegraphics[width=4.8cm]{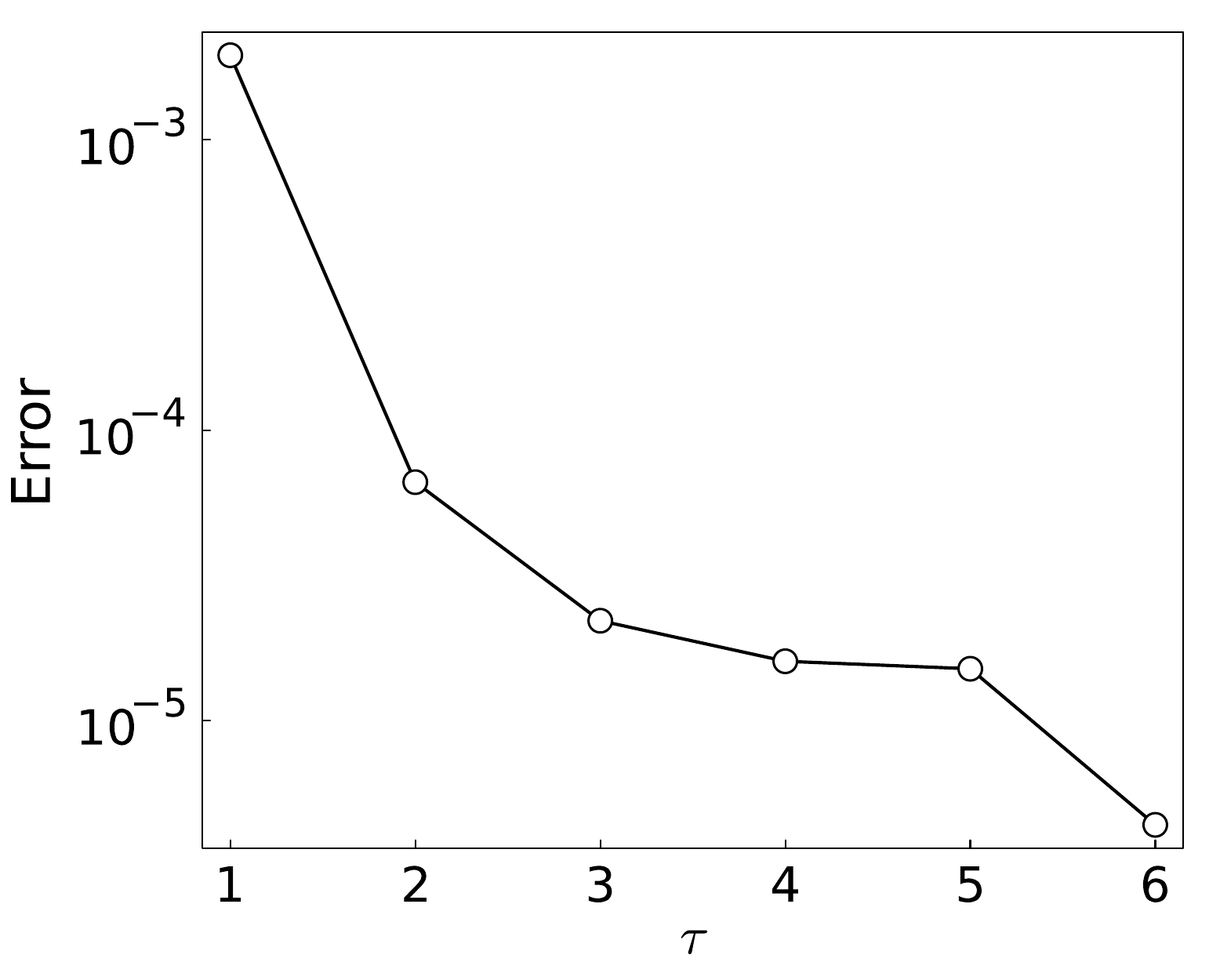} &
\includegraphics[width=4.8cm]{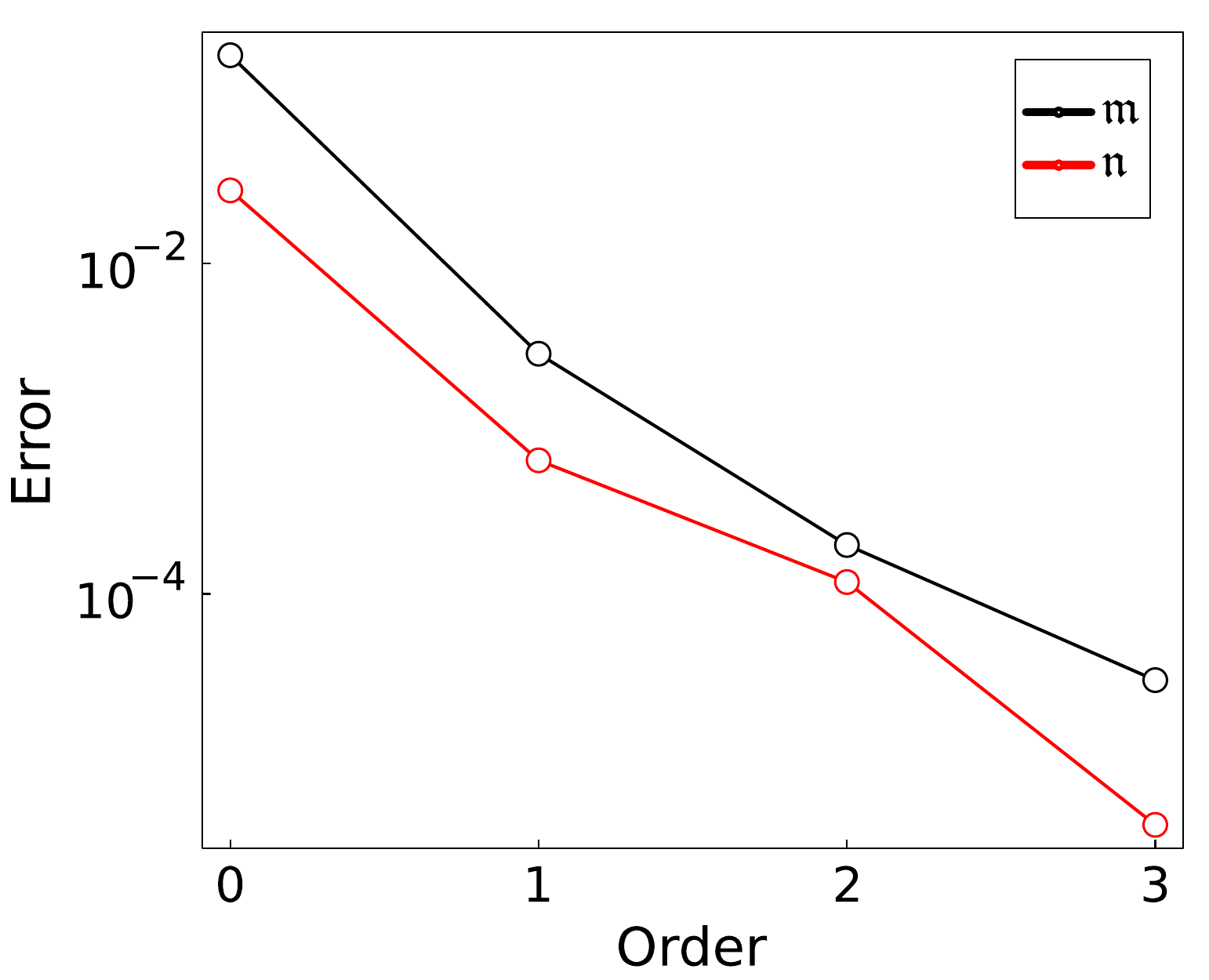} 
\end{tabular}
\setlength{\abovecaptionskip}{0pt} 
\caption{
Relative errors for the 6 electron eigenvalues closest to the Fermi energy along $\LS$. The upper pictures show the errors of the simplified TBG, where the first one is for the convergence in hopping truncation $\tau$, and the second one is for the convergence in expansion orders $\iip,\ijp$ with $\tau_0=6$. Similarly for the lower pictures, but for the Wannierized TBG, and $\tau_0=2$ for the second one.
}
\label{fig:error}
\end{figure}

We further show the electronic band structures along $\LS$ for the momentum space model $\eH$, the BM model $\eH^{(1,0,1)}$, and the high accuracy model $\eH^{(\iip,\ijp,\tau)}$ (with parameters chosen as above) respectively in Figure \ref{fig:bandstruct}. For the simplified TBG, the momentum space model gives an extremely flat band very close to 0 energy. The BM model gives a slightly less flat band with qualitatively different bands from the momentum space model at the $\Gamma$ point. This indicates that the BM model of the simplified TBG can not accurately capture these characteristics at the magic angle. By increasing both the intralayer and interlayer expansion orders by one and the hopping truncation parameter by two, the band structure is almost identical to the exact result. Therefore, the parameter choices for the high accuracy continuum model are reasonable. 
For the Wannierized TBG, the band structure of the BM model is more identical to the exact result except for a slight shift towards 0. This explains why we can use smaller corrections to obtain the accurate mode for the Wannierized TBG compared to the simplified TBG.


\begin{figure}[htb!]
\centering
\begin{tabular}{cc}
{\includegraphics[width=4.8cm]
{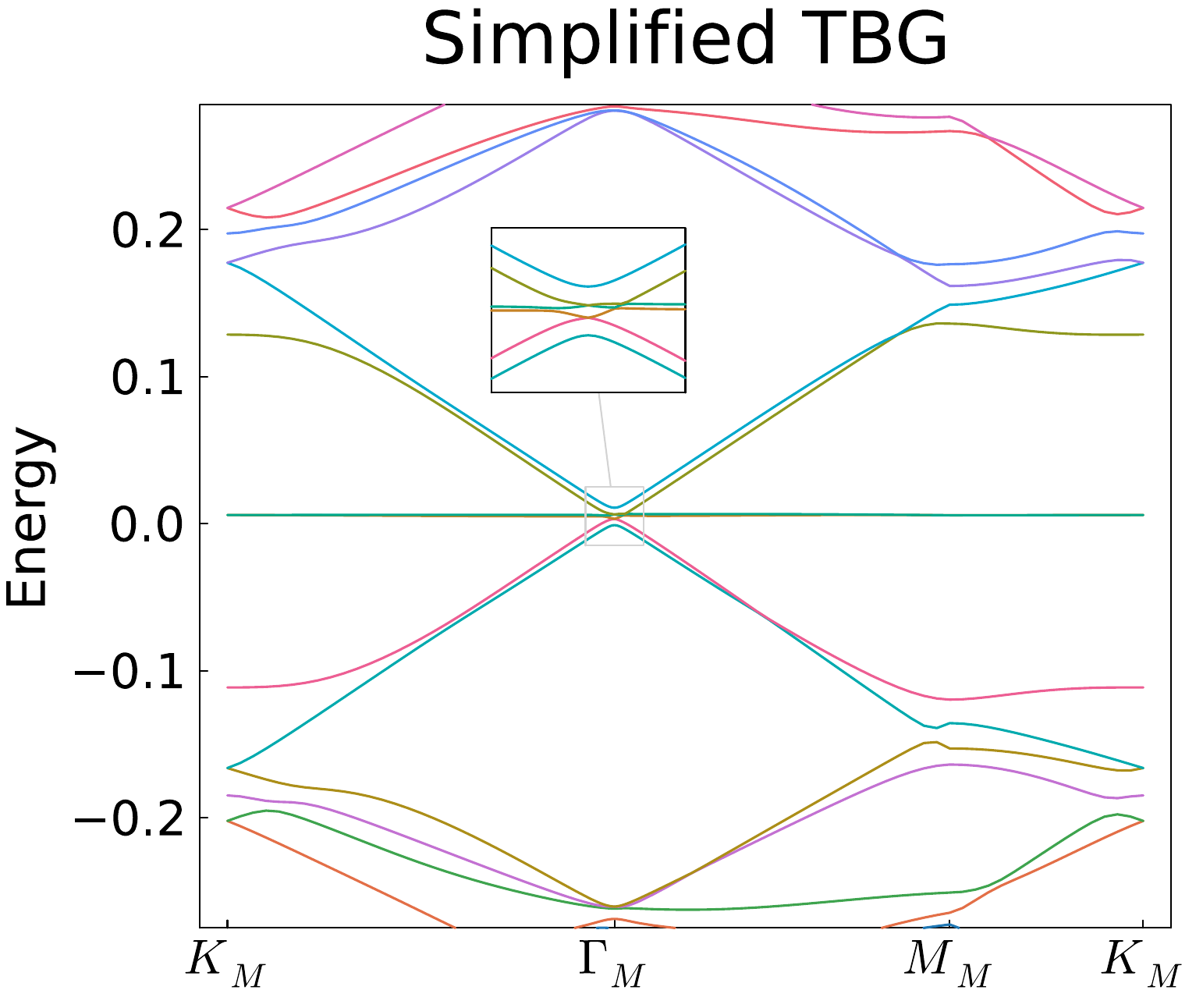}} & 
\includegraphics[width=4.8cm]
{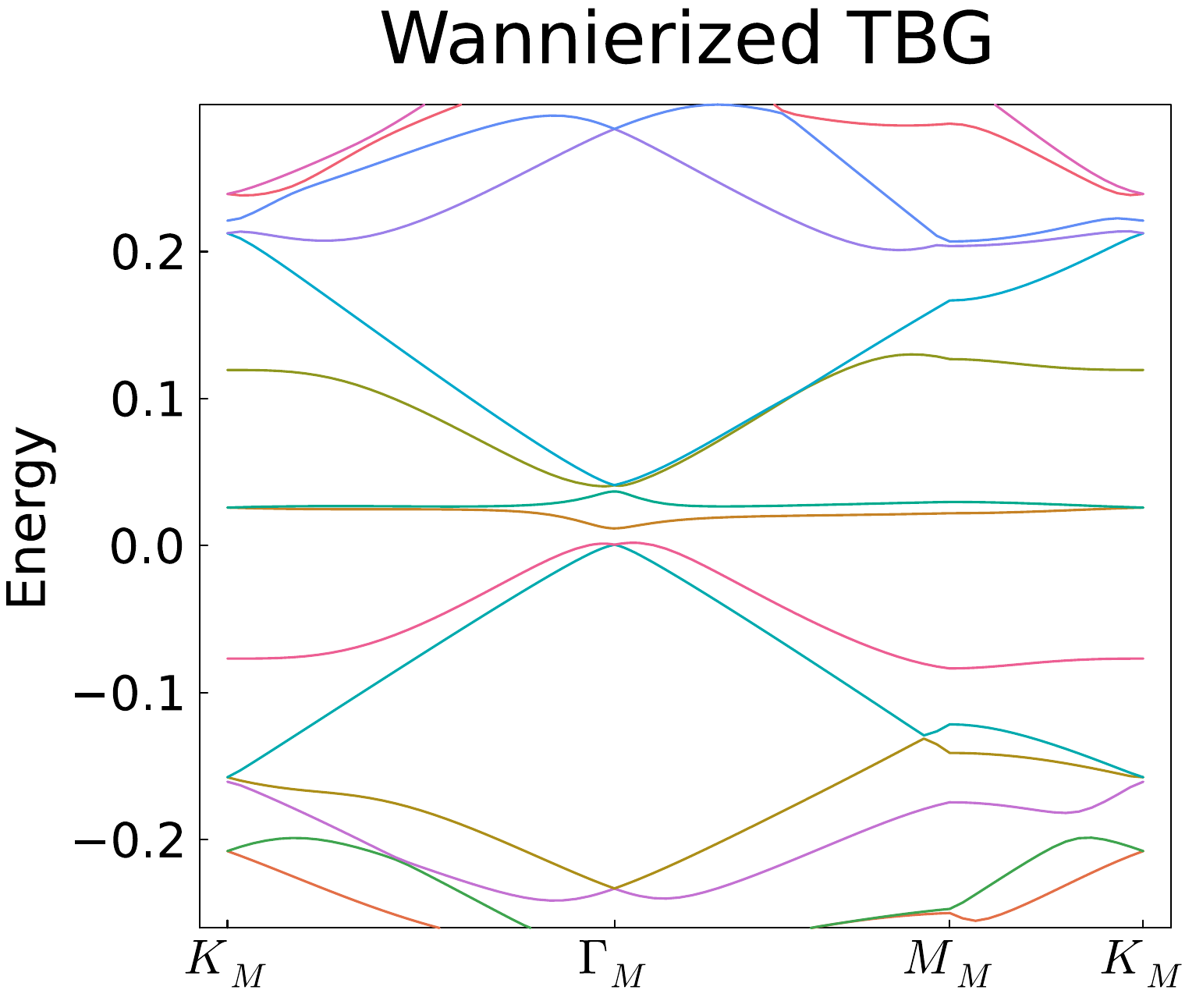} \\
\includegraphics[width=4.8cm]
{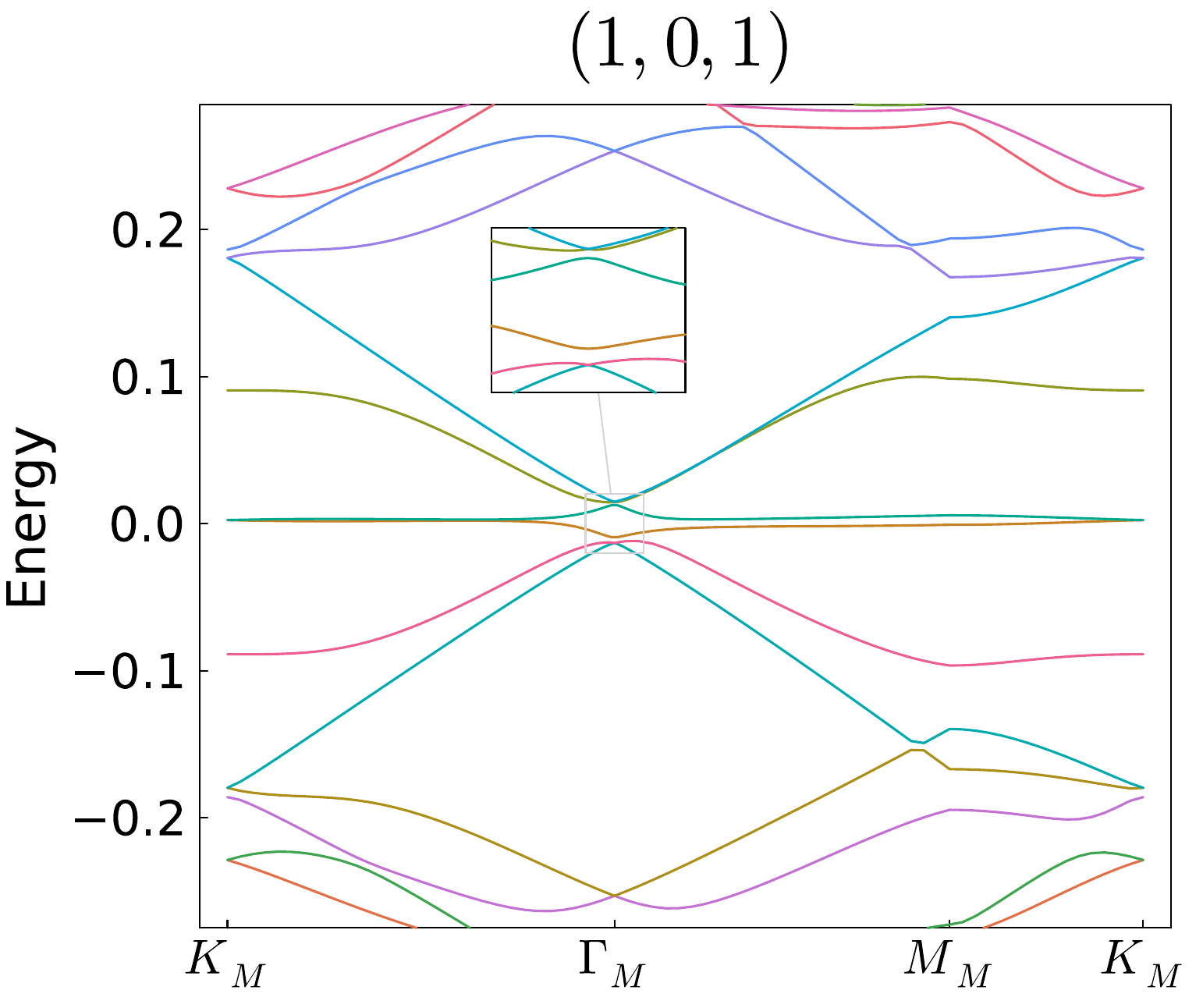} &
\includegraphics[width=4.8cm]
{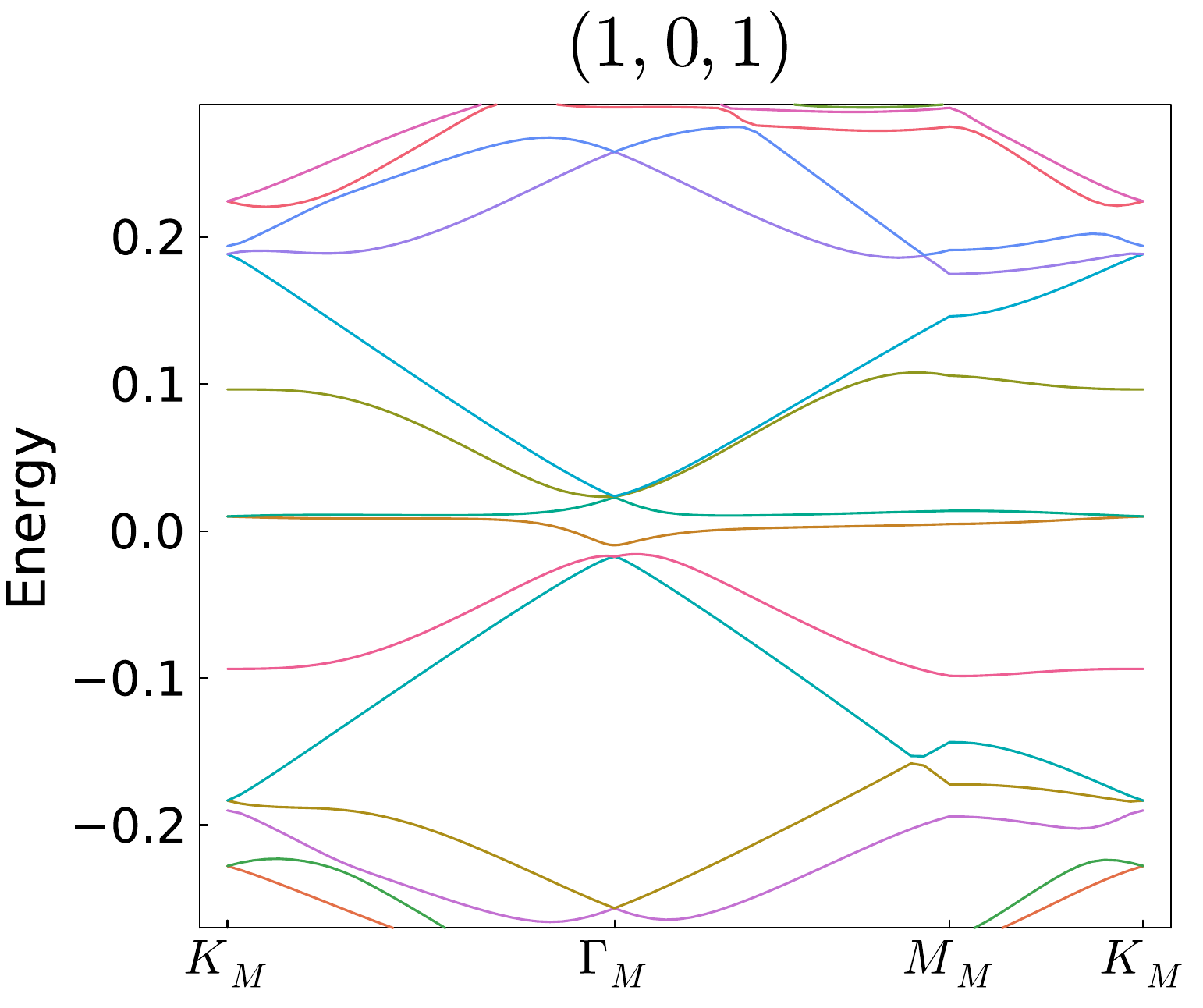} \\
\includegraphics[width=4.8cm]
{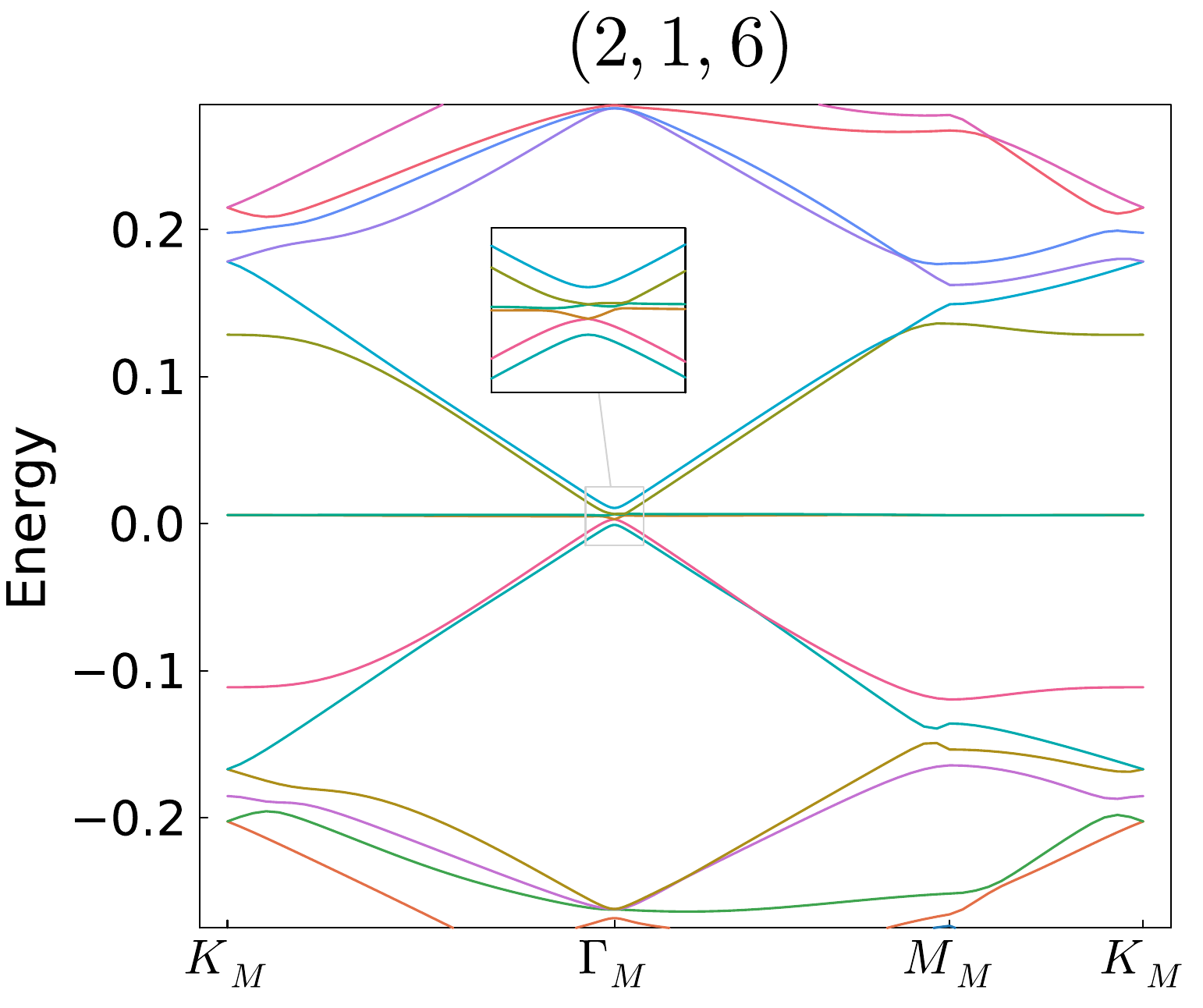} &
\includegraphics[width=4.8cm]
{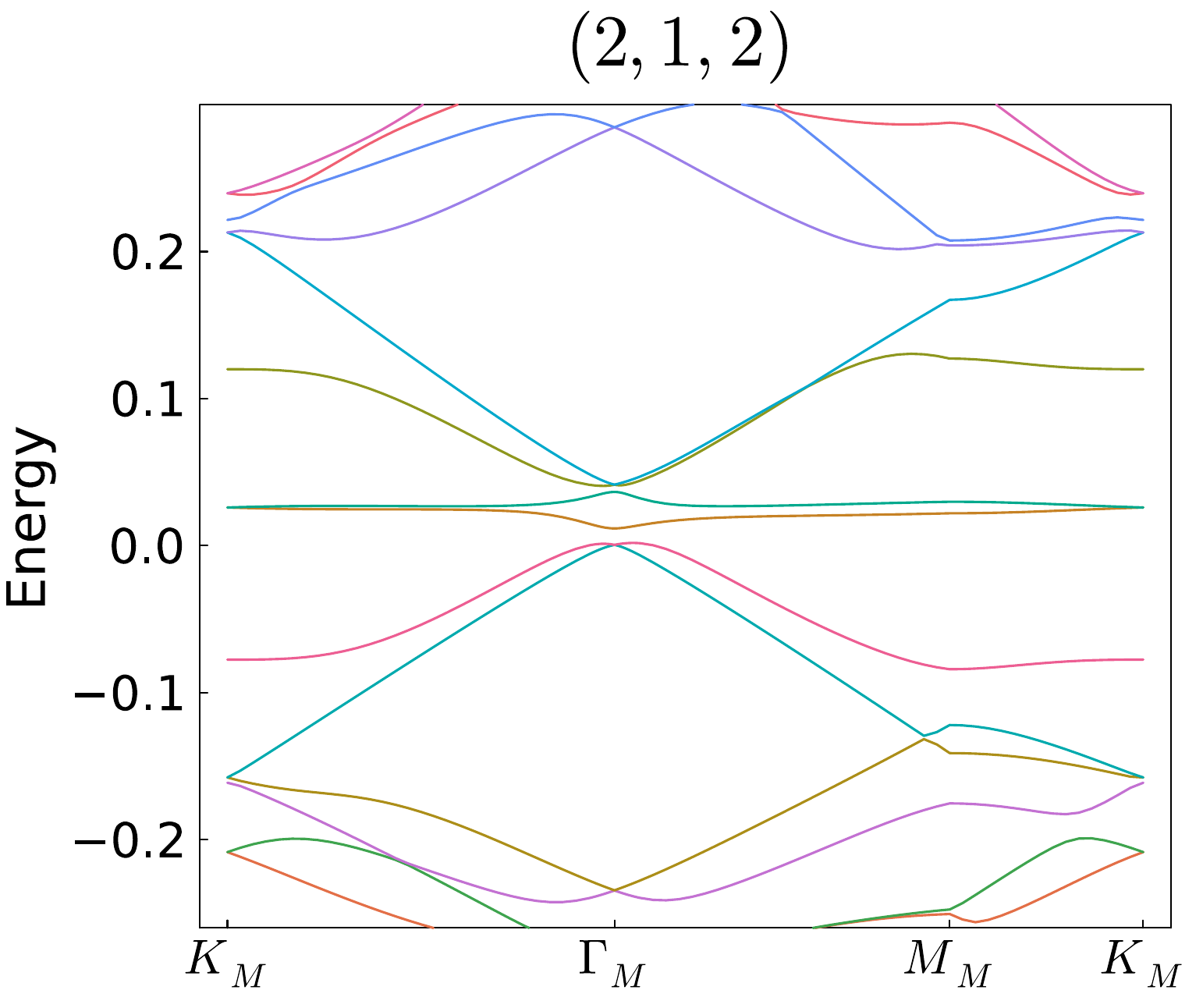} 
\end{tabular}
\setlength{\abovecaptionskip}{0pt} 
\caption{
Magic angle $1.1^\circ$ electronic band structure along $\LS$ for three models,
the momentum space model (top), the BM model (middle), and the high accuracy continuum model (bottom). The first column shows the band structures for the simplified TBG, while the second shows that of the Wannierized TBG. The momentum axes are labeled in terms of the high-symmetry points of the reciprocal lattice of the moir\'{e} supercell, not the monolayer cells.
}
\label{fig:bandstruct}
\end{figure}

\medskip


\medskip

\appendix

\section{Proof of Lemma \ref{lemma:convergence:eigenvalues}}
\label{proof:error:eigs}
Since the Taylor series of the hoppings are
\begin{align*}
    \bthop(\xi) &= \sum_{|\beta|=0}^\infty\frac{D^\beta\bthop(\xi_0)}{\beta !}(\xi-\xi_0)^\beta\\
    &=\sum_{|\beta |=0}^\infty\frac{(\xi-\xi_0)^\beta }{\beta !}\sum_{R_j\in\R_j}(-i(\tau_\sigma-\tau_{\sigma'}+R_j))^\beta e^{-i\xi_0\cdot(\tau_\sigma-\tau_{\sigma'}+R_j)} h^{jj}_{\sigma\sigma'}(R_j),
\end{align*}
and 
\begin{align*}
    \interh(\xi) &= \sum_{|\beta|=0}^\infty\frac{D^\beta\interh(\xi_0)}{\beta !}(\xi-\xi_0)^\beta\\
    &=\sum_{|\beta |=0}^\infty\frac{(\xi-\xi_0)^\beta }{\beta !}\int_{\RR^2}(-ix)^\beta e^{-i x\cdot \xi_0} {h^{12}_{\sigma\sigma'}}(x)\; dx.
\end{align*}
By Assumption \ref{assumption:interhop}, \ref{assumption:intrahop} and the Cauchy–Schwarz inequality, we have 
\begin{align*}
    \big|\bthop(\xi)\big| &\leq\sum_{R_j\in\R_j}|h^{jj}_{\sigma\sigma'}(R_j)|\sum_{|\beta|=0}^\infty\frac{|(\xi-\xi_0)^\beta(\tau_\sigma-\tau_{\sigma'}+R_j)^\beta|}{\beta!}\\
    &\lesssim  \sum_{R_j\in\R_j}e^{(|\xi-\xi_0|-\gamma_j)|R_j|},
\end{align*}
and
\begin{align*}
    \big|\interh(\xi) \big|
    &\leq\int_{\RR^2}\sum_{|\beta|=0}^\infty\frac{|(\xi-\xi_0)^\beta x^\beta|}{\beta!}{|h^{12}_{\sigma\sigma'}}(x)|\; dx\\
    &\lesssim  \int_{\RR^2}e^{(|\xi-\xi_0|-\gamma_{12})|x|}\; dx.
\end{align*}
We can therefore obtain that the Taylor series converge absolutely for 
\[
\big\{q\in\RR^2, G\in\R_2^*: |q-\tK_1+\Gjm(G)|<\min\{\gamma_1,\gamma_{12}\}\big\}
\]
and
\[
\big\{q\in\RR^2, G\in\R_1^*: |q-\tK_2+\Gjm(G)|<\gamma_2\big\}.
\]
Similarly, by the Taylor's theorem, we have the following estimates for the remainders,
\begin{align*}
\sum_{|\beta|=\iip+1}&\frac{1}{\beta!}
\max_{\xi'}\big|D^\beta\bthop(\xi')(\xi-\xi_0)^\beta\big|\\[1ex] 
\lesssim&\,
\sum_{k=\iip+1}^\infty\frac{|\xi-\xi_0|^k}{k!}\sum_{R_j\in\R_j} e^{-\gamma_j|R_j|}{|\tau_\sigma-\tau_{\sigma'}+R_j|^k}
\\[1ex]
\leq&\,\frac{1}{|\Gamma_j|}\sum_{k=\iip+1}^\infty \frac{|\xi-\xi_0|^k}{k!}\int_{\RR^2}|\tau_\sigma-\tau_{\sigma'}+x|^ke^{-\gamma_j|x|}\;dx\\[1ex]
\leq&\,\frac{e^{\gamma_j|\tau_\sigma-\tau_{\sigma'}|}}{|\Gamma_j|}\sum_{k=\iip+1}^\infty \frac{|\xi-\xi_0|^k}{k!}
    \int_0^{2\pi}\int_0^\infty r^{k+1}e^{-\gamma_jr}\;dr\;d\theta\\[1ex]
    =&\,\frac{2\pi e^{\gamma_j|\tau_\sigma-\tau_{\sigma'}|}}{\gamma_j|\Gamma_j|}\sum_{k=\iip+1}^\infty (k+1)\left(\frac{|\xi-\xi_0|}{\gamma_{j}}\right)^k \\[1ex]
    \lesssim&\, (\iip+1)\left(\frac{|\xi-\xi_0|}{\gamma_{j}}\right)^{\iip+1},
\end{align*}
and 
\begin{align*}
    \sum_{|\beta|=\ijp+1}^\infty\frac{1}{\beta!}&
\max_{\xi'}|D^\beta\interh(\xi')(\xi-\xi_0)^\beta|\\[1ex]
\lesssim&\,\sum_{k=\ijp+1}^\infty\frac{|\xi-\xi_0|^k}{k!}\int_{\RR^2}|x|^ke^{-\gamma_{12}|x|}\; dx\\[1ex]
=&\, \sum_{k=\ijp+1}^\infty\frac{|\xi-\xi_0|^k}{k!}\int_0^{2\pi} \int_0^\infty r^{k+1}e^{-\gamma_{12}r}\;dr\;d\theta\\[1ex]
=&\, 
\frac{2\pi}{\gamma_{12}}\sum_{k=\ijp+1}^\infty(k+1)\left(\frac{|\xi-\xi_0|}{\gamma_{12}}\right)^k\\[1ex]
\lesssim &\, (\ijp+1)\left(\frac{|\xi-\xi_0|}{\gamma_{12}}\right)^{\ijp+1} .
\end{align*}
Let $\Delta H:=\trcH(q)-\tpH(q)$, for $q, G$ belonging to the above subsets, we can obtain
\begin{align*}
&|(\Delta H_{jj})_{G\sigma,G\sigma'}| \leq C(\iip+1)(\gamma_j^{-1})^{\iip+1}|q-\tK_j+\Gjm(G)|^{\iip+1},\\[1ex]
&|(\Delta H_{12})_{G\sigma,G'\sigma'}|\leq C(\ijp+1)(\gamma_{12}^{-1})^{\ijp+1}|q-\tK_1+\Gjm(G)|^{\ijp+1},
\end{align*}
where the constant $C$ is independent of $\iip,\ijp, \tau$.

Note that when $q\in \tK + \mBZ$ and $G\sigma\in\Omega_r^*(q)$, we have $|q-\tK_j+\Gjm(G)|\leq r_\Sigma+r=:r'$. 
Then we view the matrix as blocks with respect to reciprocal lattices, we can have
\[
|(\Delta H_{jj})_{G,G}| \leq{C(\iip+1) ({\gamma_j^{-1}}r')^{\iip+1}}\qquad{\rm and}\qquad
|(\Delta H_{12})_{G,G'}|\leq{C(\ijp+1)(\gamma_{12}^{-1} r')^{\ijp+1}}.
\]
Let $R_G$ be the sum of the absolute values of the non-diagonal entries in the 
$G$-th row,
\[
R_G:=\sum_{G'\neq G}|\Delta H_{G,G'}|\leq\sum _{G'}|(\Delta H_{12})_{G,G'}|\delta_{\I(G)+\I(G')\in \B_\tau}\leq {{C\# \B_\tau(\ijp+1)(\gamma_{12}^{-1} r')^{\ijp+1}}}.
\]
By the Gershgorin circle theorem , we can obtain
\begin{align*}
    \|\Delta H\|_{\op}& \leq\max_{G} \big(|\Delta H_{G,G}| + R_G\big) \leq C\Big((\iip+1) (\max_{j={1,2}}{\gamma_j^{-1}}r')^{\iip+1} + \# \B_\tau(\ijp+1)(\gamma_{12}^{-1} r')^{\ijp+1}\Big).
\end{align*}

\section{Proof of Lemma \ref{lemma:continuum model}}
\label{compt:MS:GBM}
By the definition of $\mE_{q,j}$, and the matrix elements \eqref{intra:GBM} and \eqref{inter:GBM}, for $\mG\sigma\in\Omega^*_{1{\rm M}}$, we have
\begin{align*}
[\bmH(q)\mE_q\bw]^1_{\mG\sigma}&=\sum_{\mG'\sigma'\in\Omega^*_{1{\rm M}}}[\widetilde{H}^{(\iip)}_{11}(q)]_{\mG\sigma,\mG'\sigma'}e^{-i\mathfrak{G}_1(\mG')\cdot\tau_{\sigma'}}\bw_{1{\sigma'}}(q-\tK_1+\mG')\\[1ex]
    &\quad+\sum_{\mG'\sigma'\in\Omega^*_{2{\rm M}}}[\widetilde{H}^{(\ijp,\tau)}_{12}(q)]_{\mG\sigma,\mG'\sigma'}e^{-i\mathfrak{G}_2(\mG')\cdot\tau_{\sigma'}}\bw_{2\sigma'}(q-\tK_2-\mG')\\[1ex]
    &=e^{-i\mathfrak{G}_1({\mG})\cdot\tau_{\sigma}}\sum_{\sigma'\in\A_1}\cP^{\iip}_{1,\sigma\sigma'}\big(q-\tK_1+\mG\big)\bw_{1\sigma'}(q-\tK_1+\mG)\\[1ex]
    &\quad+e^{-i\mathfrak{G}_1({\mG})\cdot\tau_{\sigma}}\sum_{\sigma'\in\A_2,\,\I(\GG_{\rm{M}})\in\B_\tau}e^{i\mathfrak{G}_1(\GG_{\rm{M}})\cdot\tau_{\sigma}}e^{-i\mathfrak{G}_2({\GG_{\rm{M}}})\cdot\tau_{\sigma'}}\\[1ex]
    &\qquad  \cU^{\ijp}_{\mathfrak{G}_1(\GG_{\rm M}),\sigma\sigma'}(q-\tK_1+\mG)\bw_{2\sigma'}(q-\tK_1+\mG+\tK_1-\tK_2-\GG_{\rm{M}})\\[1ex]
    &=\mE_{q,1}\Bigg\{\sum_{\sigma'\in\A_1}\cP^{\iip}_{1,\sigma\sigma'}(\cdot)\bw_{1\sigma'}+\sum_{\sigma'\in\A_2,\,\I(\GG_{\rm{M}})\in\B_\tau}e^{i\mathfrak{G}_1(\GG_{\rm{M}})\cdot\tau_{\sigma}}e^{-i\mathfrak{G}_2(\GG_{\rm{M}})\cdot\tau_{\sigma'}} \\
    &\qquad \qquad \cU^{\ijp}_{\mathfrak{G}_1(\GG_{\rm{M}}),\sigma\sigma'}(\cdot)\bw_{2\sigma'}(\cdot+\tK_1-\tK_2-{\GG_{\rm M}})\Bigg\}_{\mG\sigma}.
\end{align*}
Therefore, we can deduce that
\begin{align*}
    &[H_{jj}^{(\iip)}(q)]_{\sigma\sigma'}=\cP_{j,\sigma\sigma'}^\iip(q),\\
    &
    [H_{12}^{(\ijp,\tau)}(q)]_{\sigma\sigma'}=\sum\limits_{\I(\mG)\in\B_\tau}e^{i\mathfrak{G}_1(\mG)\cdot\tau_{\sigma}}e^{-i\mathfrak{G}_2(\mG)\cdot\tau_{\sigma'}} \cU^{\ijp}_{\mathfrak{G}_1(\mG),\sigma\sigma'}(q)T_{\tK_1-\tK_2-\mG}.
\end{align*}

\section{Proof of Theorem \ref{thm:convergence:dos}}
\label{proof:error:dos}
We denote the approximate DoS of $\tpH$ by $D^{(\iip,\ijp,\tau)}_{\varepsilon,r}(E;\tilde{K})$, and then divide the error into three parts:
\begin{align*}
   \Big|D_\varepsilon(E)-\sum_{\tK\in\{K,K'\}}D^{(\iip,\ijp,\tau)}_{\varepsilon}(E;\tK)\Big|
   &\leq
\Big|D_\varepsilon(E)-\sum_{\tK\in\{K,K'\}}D_{\varepsilon,r}^{(\tau)}(E;\tK)\Big|\\
&\quad+ \sum_{\tK\in\{K,K'\}}\Big|D_{\varepsilon,r}^{(\tau)}(E;\tK)-D^{(\iip,\ijp,\tau)}_{\varepsilon,r}(E;\tK)\Big|\\
&\quad+\sum_{\tK\in\{K,K'\}}\Big|D^{(\iip,\ijp,\tau)}_{\varepsilon,r}(E;\tK)-D^{(\iip,\ijp,\tau)}_{\varepsilon}(E;\tK)\Big|\\[1ex]
&=: I_{1}+I_{2}+I_{3}.
\end{align*}
Since the truncation part $I_1$ has been given by Theorem \ref{thm:truncation:dos},
we only need to quantify polynomial approximation error $I_2$ and extension error $I_3$. For simplicity, we define the following notations to distinguish the intralayer part and
interlayer part for a sheet-wise decomposed matrix $A=\begin{pmatrix} A_{11} & A_{12} \\ A_{12}^\dagger & A_{22} \end{pmatrix}$,
\[
A_{\rm intra} := \begin{pmatrix} A_{11} & 0 \\ 0 & A_{22} \end{pmatrix},\qquad A_{\rm inter}:=\begin{pmatrix} 0 & A_{12} \\ A_{12}^\dagger & 0 \end{pmatrix}.
\]

For the polynomial approximation error, we have

\begin{align*}
  I_2&\leq
\nu^*\int_{\tK+\mBZ}\big|{\rm Tr}\;\gauss(E-\trcH(q))-{\rm Tr}\;\gauss(E-\tpH(q))\big|\,dq\\[1ex]
&\leq \nu^*\sum_{i=1}^{\#\Omega_r^*(q)}\int_{\tK+\mBZ} |\gauss(E-\eig)-\gauss(E-\eigtp)|\,  dq,
\end{align*}
where $\eig$ and $\eigtp$ denote the $i$-th eigenvalues of $\trcH(q)$ and $\tpH(q)$ respectively, and sorted in ascending order.
Let $V$ be the $i$-dimensional subspace of $\ell^\infty(\Omega_r^*(q))$. For any $\Psi\in V$, $\|\Psi\|_2 = 1$, we have
\[
\big|\big(\Psi,\trcH\Psi\big)-\big(\Psi,\tpH\Psi\big)\big|\leq\|\Delta H\|_\op,
\]
where $\Delta H:=\trcH-\tpH$. Hence
\[
\big(\Psi,\tpH\Psi\big)-\|\Delta H\|_\op\leq\big(\Psi,\trcH\Psi\big)\leq\big(\Psi,\tpH\Psi\big)+\|\Delta H\|_\op.
\]
By the min-max theorem, we obtain that for any $i\in\{1,\cdots,\#\Omega_r^*(q)\}$ ,
\[
|\eig-\eigtp|\leq\|\Delta H\|_\op.
\]
We thus have 
\begin{align*}
    |\gauss(E-\eig)-\gauss(E-\eigtp)|&\leq\max_{x}\big|\gauss'(x)\big| |\eig-\eigtp|\lesssim \varepsilon^{-2}\|\Delta H\|_\op.
\end{align*}
In addition, we note that some error is also controlled by the Gaussian tail. 
Let $\lambda_i(q)$ be the $i$-th eigenvalue of $\eH^{(\tau)}_{r,{\rm intra}}$, we similarly have 
\begin{equation*}
    |\lambda_i(q)- \eig|\leq \|\trcH-\eH^{(\tau)}_{r,{\rm intra}}\|_{\op}<\frac{\eta}{2},
\end{equation*}
where $\eta$ is definded in \eqref{couple_strength}.
Moreover, we separate the eigenvalues by energy windows. Consider an increasing collection of energy $ E_k = \Sigma+\frac{k+2}{2}\eta,\; k=-1,\cdots ,n+1$, satisfying 
$E_{n-1}<\|\eH^{(\tau)}_{r,{\rm intra}}\|_{\op}\leq E_n$. Then for $\lambda_i(q)\in B_{E_k}\setminus B_{E_{k-1}}, \; 1\leq k \leq n$, we have $\eig\in B_{E_{k+1}}\setminus B_{E_{k-2}}$.
Therefore, for $i$ such that $\lambda_i(q)\in B_{E_k}\setminus B_{E_{k-1}}$, there is a constant $\gamma_g$ such that
\begin{equation}
\label{gauss_tail}
|\gauss(E-\eig)|\lesssim \varepsilon^{-1}e^{-\gamma_g k^2\varepsilon^{-2}}.
\end{equation}
And by the conical structure of the monolayer band structure, we have 
\begin{equation}
\label{window_size}
\#\{i:\lambda_i(q)\in B_{E_k}\setminus B_{E_{k-1}}\} \sim k.
\end{equation}
Then by \eqref{gauss_tail} and \eqref{window_size}, we obtain
\begin{align*}
    \sum_{i,\lambda_i(q) \notin \ER}|\gauss(E-\eig)|&=\sum_{k=1}^n \sum_{i,\lambda_i(q)\in B_{E_k}\setminus B_{E_{k-1}}}|\gauss(E-\eig)|\\
    &\lesssim \varepsilon^{-1}\sum_{k=1}^nke^{-\gamma_g k^2\varepsilon^{-2}}\\
    &\lesssim \varepsilon^{-1}\int_{1}^nke^{-\gamma_g k^2\varepsilon^{-2}}\;dk\\
    &\lesssim \varepsilon e^{-\gamma_g\varepsilon^{-2}}.
\end{align*}
We also note that
\[ \#\{i:\lambda_i(q) \in \ER\}\sim\#\Omega_{0}^*(\tK).
\]
The analysis is similar for $\eigtp$.
Then by Lemma \ref{lemma:convergence:eigenvalues}, we obtain
\[
I_2\lesssim \varepsilon^{-2}\max_{\tK\in\{K,K'\}}\#\Omega_{0}^*(\tK)\Big((\iip+1) \big(\max_{j={1,2}}{\gamma_j^{-1}}(r_\Sigma+r)\big)^{\iip+1} + \#\B_\tau(\ijp+1)\big(\gamma_{12}^{-1} (r_\Sigma+r)\big)^{\ijp+1}\Big) + \varepsilon e^{-\gamma_g\varepsilon^{-2}}.
\]

We now estimate $I_3$. Since $\tpH$ and $\bmH$ have same hopping truncation and expansion orders, we could view $\iip,\ijp$ and $\tau$ as fixed parameters. 
For simplicity of notations, we denote for matrices $B(q)$
\[
T(B)= \nu^*\int_{\tK+\mBZ}B(q)\;dq,
\]
so in particular,
\[
D^{(\iip,\ijp,\tau)}_{\varepsilon}(E;\tK) = \Tr\;T\bigl(\gauss(E-\bmH)\bigr).
\]
Let $\CC$ be a contour around the spectrum $\bmH$ such that $d(\CC,\bmH)\in(\varepsilon,2\varepsilon)$. If the spectrum has gaps larger than $\varepsilon$, then $\CC$ would not be a simple curve in the complex plane but a union of one per ungapped interval of spectrum.
We observe that
\[
I_3=\frac{1}{2\pi}\bigg|\oint_\CC \gauss(E-z)\Big(\Tr\;T\bigl((z-\bmH)^{-1}\bigr)-\Tr\;T\bigl((z-\tpH)^{-1}\bigr)\Big)dz\bigg|.
\]
We use the ring decomposition technique from \cite{Massatt2023} to find the bound of $I_3$. We define a new energy range 
\[
\tilde{\Sigma}:=\Sigma + \|\eH^{(\iip,\ijp,\tau)}_{r,{\rm intra}}- \eH^{(\tau)}_{r,{\rm intra}}\|_{\op},
\]
a new enlarge strength
\[
\tilde{\eta} := (2+\alpha)\|\bmH_{\rm inter}\|_{\op}.
\]
and a new radius $\tilde{r}$ satisfying
$\Gamma_j^*(B_{\tilde{\Sigma}+\tilde{\eta}})+B_{\tilde{r}}(0)\subseteq\Gamma_j^*(\ER)+B_r(0)$ and $\Gamma_j^*(B_{\tilde{\Sigma}+\tilde{\eta}})+B_{\tilde{r}+\delta}(0)\nsubseteq\Gamma_j^*(\ER)+B_r(0)$ for any $\delta>0$.
Note that $\tilde{r}\approx r$ for $\iip,\ijp$ large enough.
Consider an increasing collection of radii $r_0,\cdots, r_n, r_{n+1}, r_{n+2}\cdots$ such that $r_0=0$ and $r_n=\tilde{r}$. For $q\in \tK + \mBZ$, we write $H:=\bmH(q)$, and use \eqref{finite basis}, \eqref{inclusion} to define the following decomposition: 
\begin{align*}
    &U_0 = I\big(\Omega_{r_0}^*(q,B_{\tilde{\Sigma}+\tilde{\eta}})\big),\\
    &U_j = I\big(\Omega_{r_j}^*(q,B_{\tilde{\Sigma}+\tilde{\eta}})\big)\setminus I\big(\Omega_{r_{j-1}}^*(q,B_{\tilde{\Sigma}+\tilde{\eta}})\big), ~j>0,\\
    &J_j=J_{\Omega^*\gets U_j},\\
    &H_{ij}=J_i^*H J_j.
\end{align*}
We choose $n$ and $r_j=j\tilde{r}/n$ such that $H_{ij}=0$ if $|i-j|>1$, then the decomposition gives a ``nearest neighbor" decomposition. Since the sites are coupled in such a way that $\I(G)+\I(G')\in \B_\tau$ or $\I(\mG+\mG')\in \B_\tau$, this can be easily achieved by choosing $n$ such that the index distance between the sites of two neighboring rings is proportional to 
\[
C_\tau:=\max_{\pmb{n}\in\B_\tau}|\pmb{n}|.
\]
So we have the number of entries of each ring 
\begin{equation}
    \label{ring_size}
    \#U_0\sim \#\Omega^*_{0}(\tK),\qquad \#U_j\sim jC_\tau^2,~j>0.
\end{equation}
We observe that the Hamiltonian is
\[
H = \begin{pmatrix}
    H_{00}& H_{01}&0&0&\cdots\\
    H_{10}& H_{11}&H_{12}&0&\cdots\\
     0&H_{21}& H_{22}&H_{23}&\cdots\\
     0&0&H_{32}& H_{33}&\cdots\\
     \vdots&\ddots&\ddots&\ddots&\ddots
\end{pmatrix}.
\]
We let $\Hrl{i}{j}$ be the matrix restricted to the rings $i$ through $j$ for $i<j$, and $\Jrl{i}{j}$ be the corresponding inclusion. We also use the resolvent notations:
\[
\Rrl{i}{j}=(z-\Hrl{i}{j})^{-1},\qquad R_j = (z-H_{jj})^{-1},\qquad R=(z-H)^{-1}.
\]
Then by \eqref{ring_size},
\begin{align}
\label{I3}
    I_3 &\leq \frac{1}{2\pi}\sum_{k=0}^\infty\bigg|\oint_\CC \gauss(E-z)\big(\Tr\;T(J_k^*RJ_k)-\Tr\;T(J_k^*\Rrl{0}{n} J_k)\big)dz\bigg |\\
    \nonumber
    &\lesssim{\#\Omega^*_{0}(\tK)}\bigg\|\oint_\CC \gauss(E-z)\big(T(J_0^*RJ_0)-T(J_0^*\Rrl{0}{n} J_0)\big)\;dz\bigg\|_{\op}\\
    \nonumber
    &\quad+{C_\tau^2}\sum_{k=1}^\infty k\bigg\|\oint_\CC \gauss(E-z)\big(T(J_k^*RJ_k)-T(J_k^*\Rrl{0}{n} J_k)\big)\;dz\bigg\|_{\op}.
\end{align}
It suffices to show that for $\forall q\in\tK+\mBZ$,
\begin{align}
    \label{I3_1}
    \bigg\|\oint_\CC \gauss(E-z)(J_0^*RJ_0-J_0^*\Rrl{0}{n} J_0)\;dz\bigg\|_{\op}\lesssim \varepsilon^{-4}e^{-\gamma_m r}+\varepsilon^{-1} e^{-\gamma_g\varepsilon^{-2}},
\end{align}
\begin{align}
    \label{I3_2}
   \sum_{k=1}^\infty k\bigg\|\oint_\CC \gauss(E-z)(J_k^*RJ_k-J_k^*\Rrl{0}{n} J_k)\;dz\bigg\|_{\op}\lesssim  \varepsilon^{-4}(r+1) e ^{-\gamma_m r} + \varepsilon e^{-\gamma_g \varepsilon^{-2}}.
\end{align}
In the following, we will omit the notation $q$, and consider two cases for $k\in\{0,\cdots,\infty\}$.

{\bf Case 1, $\mathbf{k\leq n/2}$:}
 We divide $\CC$ into two regions,
\begin{equation}
    \label{contour}
    \CC_k^{+}=\Big\{z\in\CC:{\rm Re}(z)\in B_{\Sigma_k+\tilde{\eta}'}\Big\},\qquad\CC_k^{-}\in \CC\setminus\CC_k^{+},
\end{equation}
where 
\begin{align*}
  \Sigma_k:= \Sigma + \|\Jrl{0}{k}^*\bmH_{\rm intra}\Jrl{0}{k}\|_{\op}-\|J_0^*\bmH_{\rm intra}J_0\|_{\op},\qquad\tilde{\eta}' := \frac{\alpha}{2}\|\bmH_{\rm inter}\|_{\op}.
\end{align*}
We have 
\begin{align*}
    & \oint_\CC \gauss(E-z)(J_k^*RJ_k-J_k^*R_{0\leftrightarrow n} J_k)\;dz\\[1ex]
    = &\oint_{\CC_k^{+}} \gauss(E-z)(J_k^*RJ_k-J_k^*R_{0\leftrightarrow n} J_k)dz+\oint_{\CC_k^{-}} \gauss(E-z)(J_k^*RJ_k-J_k^*\Rrl{0}{n} J_k)\;dz.
\end{align*}
We observe that for $z\in\CC_k^{-}$, there is a $\gamma_g>0$ such that 
$$|E-z|\geq \sqrt{2\gamma_g}(k+1),$$
so the second term is bounded by 
\begin{align}
\label{proof:bound:I2:minusC}
    \nonumber\bigg\|\oint_{\CC_k^{-}} \gauss(E-z)(J_k^*RJ_k-J_k^*\Rrl{0}{n} J_k)\;dz\bigg\|_{\op}&\lesssim \varepsilon^{-1}\oint_{\CC_k^{-}}|\gauss(E-z)|\;|dz|\\
    \nonumber&\lesssim \varepsilon^{-2}\int_{|E-t|\geq\sqrt{2\gamma_g}(k+1)} e^{-|E-t|^2/2\varepsilon^2}\;dt\\
    &\lesssim \varepsilon^{-1} e^{-\gamma_g(k+1)^2\varepsilon^{-2}}.
\end{align}
Hence, it suffices to prove for arbitrary $z\in\CC_k^{+}$, there is a $\lambda>0$,
\begin{equation}
    \label{proof:bound:I2:plusC}
    \|J_k^*RJ_k-J_k^*R_{0\leftrightarrow n} J_k\|_{\op}\lesssim \varepsilon^{-3}e^{-\lambda (n-k)}.
\end{equation}
Note that we can directly obtain \eqref{I3_1} by taking $k=0$ for \eqref{proof:bound:I2:minusC} and \eqref{proof:bound:I2:plusC}, and using $n\sim r$.
To obtain \eqref{proof:bound:I2:plusC}, we 
use the Schur complement for the ring decomposition
\begin{align*}
    &\|J_k^*RJ_k-J_k^*\Rrl{0}{n} J_k\|_{\op}\\[1ex]
    =~&\|J_k^*\Rrl{0}{n}\Jrl{0}{n}^*H\Jrl{n+1}{\infty} \Jrl{n+1}{\infty}^*R\Jrl{n+1}{\infty} \Jrl{n+1}{\infty}^*H\Jrl{0}{n}\Rrl{0}{n}J_k\|_\op\\[1ex]
    \lesssim~&\varepsilon^{-2}\|J_k^*\Rrl{0}{n}J_n\|_\op.
\end{align*}
The last inequality is found by noting $\Jrl{0}{n}^*H\Jrl{n+1}{\infty} \Jrl{n+1}{\infty}^*$ only couples ring $n$ on the left to ring $n+1$ on the right.
We then rewrite $J_k^*\Rrl{0}{n}J_n$ using the top right entry in an iterative fashion as follows: 
\begin{align}
\label{schur1}
    \nonumber
    J_k^*\Rrl{0}{n}J_n&=J_k^*\Rrl{0}{n}J_{n-1}H_{n-1,n}R_n\\[1ex]
    \nonumber
    &=J_k^*\Rrl{0}{n}J_{n-2}H_{n-2,n-1}J_{n-1}^*\Rrl{n-1}{n}J_{n-1}H_{n-1,n}R_n\\[1ex]
    \nonumber
    &~\,\vdots\\[1ex]
    &=J_k^*\Rrl{0}{n}J_k\prod_{j=k+1}^nH_{j-1,j}J_j^*\Rrl{j}{n}J_j.
\end{align}
Recalling the definition of $\Jrl{j}{n}$ and $\CC_k^+$, 
\begin{align}
\label{schur1_bound}
\nonumber
&~\|H_{j-1,j}J_j^*\Rrl{j}{n}J_j\|_\op\\[1ex]
\nonumber
    \leq&~\|H_{j-1,j} \|_\op\|(z-\Hrl{j}{n})^{-1}\|_\op\\[1ex]
    \nonumber
\leq&~\|\bmH_{\rm inter}\|_{\op}\bigg(\min|\sigma(z-\Jrl{j}{n}^*(\widetilde{H}^{(\iip,\ijp,\tau)}-\bmH_{\rm inter})\Jrl{j}{n})|-\|\bmH_{\rm inter}\|_\op\bigg)^{-1}\\[1ex]
\nonumber
\leq&~\|\bmH_{\rm inter}\|_{\op}\bigg(\min|\sigma(\Jrl{j}{n}^*\bmH_{\rm intra}\Jrl{j}{n})|-|{\rm Re}(z)|-\|\bmH_{\rm inter}\|_\op\bigg)^{-1}\\[1ex]
\leq&~\frac{1}{1+\alpha/2}.
\end{align}
Then by \eqref{schur1} and \eqref{schur1_bound}, we obtain 
\[
\|J_k^*\Rrl{0}{n}J_n\|_\op\lesssim\varepsilon^{-1}\bigg(\frac{1}{1+\alpha/2}\bigg)^{n-k}=\varepsilon^{-1}e^{-\lambda(n-k)}
\]
for some $\lambda>0$.

{\bf Case 2, $\mathbf{k>n/2}$:} We observe that
\begin{equation}
    \label{rings1}
        \oint_\CC \gauss(E-z)J_k^*R J_k\;dz
   =\oint_\CC \gauss(E-z)J_k^*(R-\Rrl{k/2}{\infty})J_k\;dz + 2\pi iJ_k^*\gauss(E-\Hrl{k/2}{\infty})J_k,
\end{equation}
\begin{equation}
    \label{rings2}
\oint_\CC \gauss(E-z)J_k^*\Rrl{0}{n} J_k\;dz
   = \oint_\CC \gauss(E-z)J_k^*(\Rrl{0}{n}-\Rrl{k/2}{n})J_k\;dz + 2\pi iJ_k^*\gauss(E-\Hrl{k/2}{n})J_k.
\end{equation}
We first prove that the first term of  \eqref{rings1} is bounded by
\begin{align*}
    \bigg\|\oint_\CC \gauss(E-z)J_k^*(R-\Rrl{k/2}{\infty})J_k\;dz\bigg\|_{\op}\lesssim \varepsilon^{-4}e^{-\lambda k/2} + \varepsilon^{-1} e^{-\gamma_g(k+1)^2\varepsilon^{-2}}.
\end{align*}
Likewise, we divide $\CC$ into $\CC_k^+$ and $\CC_k^-$ as $\eqref{contour}$, but with
\[
\Sigma_k := \Sigma + \|\Jrl{0}{k/2-1}^*\bmH_{\rm intra}\Jrl{0}{k/2-1}\|_{\op}-\|J_0^*\bmH_{\rm intra}J_0\|_{\op}.
\]
For the $\CC_k^{-}$ term, we have the same bound as for \eqref{proof:bound:I2:minusC}. Hence, it suffices to show that for $z\in\CC_k^{+}$,
\[
\|J_k^*(R-\Rrl{k/2}{\infty})J_k\|_\op\lesssim\varepsilon^{-3}e^{-\lambda k/2}.
\]
We observe that
\begin{align*}
    &\|J_k^*(R-\Rrl{k/2}{\infty})J_k\|_\op\\[1ex]
    =~&\|J_k^*\Rrl{k/2}{\infty}\Jrl{k/2}{\infty}^* H \Jrl{0}{k/2-1}\Jrl{0}{k/2-1}^*R\Jrl{0}{k/2-1}\Jrl{0}{k/2-1}^*H\Jrl{k/2}{\infty}\Rrl{k/2}{\infty}J_k\|_\op\\[1ex]
    \lesssim~&\varepsilon^{-2}\|J_k^*\Rrl{k/2}{\infty}J_{k/2}\|_{\op}.
\end{align*}
We next rewrite $J_k^*\Rrl{k/2}{\infty}J_{k/2}$ using the bottom left entry in an iterative fashion as follows: 
\begin{align*}
    &~J_k^*\Rrl{k/2}{\infty}J_{k/2}\\[1ex]
    =&~J_k^*\Rrl{k/2+1}{\infty}J_{k/2+1}H_{k/2+1,k/2}J_{k/2}^*\Rrl{k/2}{\infty}J_{k/2}\\[1ex]
    =&~J_k^*\Rrl{k/2+2}{\infty}J_{k/2+2}H_{k/2+2,k/2+1}J_{k/2+1}^*\Rrl{k/2+1}{\infty}J_{k/2+1}H_{k/2+1,k/2}J_{k/2}^*\Rrl{k/2}{\infty}J_{k/2}\\\textbf{}
    ~\vdots~\\[1ex]
    =&~J_k^*\Rrl{k}{\infty}J_k\prod_{j=k/2}^{k-1}H_{j+1,j}J_j^*\Rrl{j}{\infty}J_j.
\end{align*}
By the same argument above
\[
\|J_k^*\Rrl{k/2}{\infty}J_{k/2}\|_{\op}\lesssim\varepsilon^{-1}\bigg(\frac{1}{1+\alpha/2}\bigg)^{k/2}= \varepsilon^{-1}e^{-\lambda k/2}.
\]
We next show that 
$$\|J_k^*\gauss(E-\Hrl{k/2}{\infty})J_k\|_{\op}\lesssim\varepsilon^{-1}e^{-\gamma_g (k+1)^2\varepsilon^{-2}}.$$
We observe 
\[
\min|\sigma(\Hrl{k/2}{\infty})|\geq \min|\sigma(\Jrl{k/2}{\infty}^*\bmH_{\rm intra}\Jrl{k/2}{\infty})|-\|\Jrl{k/2}{\infty}^*\bmH_{\rm inter}\Jrl{k/2}{\infty}\|_{\op}  .
\]
 Here the spectrum of $\Jrl{k/2}{\infty}^*\bmH_{\rm intra}\Jrl{k/2}{\infty}$ can be approximated to increase linearly with $k$. The reason for this is that for large $\iip$, the spectrum of $\Jrl{k/2}{n}^*\bmH_{\rm intra}\Jrl{k/2}{n}$ is a conical structure, and $\Jrl{n+1}{\infty}^*\bmH_{\rm intra}\Jrl{n+1}{\infty}$ only contains 1 order and negligible higher order modifications. We also have $\|\Jrl{k/2}{\infty}^*\bmH_{\rm inter}\Jrl{k/2}{\infty}\|_{\op}$ is bounded, since for large $\ijp$,  
 $\|\Jrl{k/2}{n}^*\bmH_{\rm inter}\Jrl{k/2}{n}\|_{\op}$ is close to $\|\hat{H}_{\rm inter}\|_{\op}$, and $\Jrl{n+1}{\infty}^*\bmH_{\rm inter}\Jrl{n+1}{\infty}$ only contains 0 order and negligible higher order modifications.
We then obtain
\[
\min|\sigma(\Hrl{k/2}{\infty})|\gtrsim k+1,
\]
and hence
\[
\|\gauss(E-\Hrl{k/2}{\infty})\|_{\op}\lesssim \varepsilon^{-1}e^{-\gamma_g (k+1)^2\varepsilon^{-2}}.
\]
Using the same argument above, for $n/2<k\leq n$, \eqref{rings2} is similarly bounded by 
\[
\bigg\|\oint_\CC \gauss(E-z)J_k^*\Rrl{0}{n} J_k\;dz\bigg\|_{\op}\lesssim \varepsilon^{-4}e^{-\lambda k/2} + \varepsilon^{-1} e^{-\gamma_g(k+1)^2\varepsilon^{-2}}.
\]

Therefore, we have
\begin{align*}
    &\sum_{k=1}^\infty k\bigg\|\oint_\CC \gauss(E-z)(J_k^*RJ_k-J_k^*\Rrl{0}{n} J_k)\;dz\bigg\|_{\op}\\[1ex]
    &\qquad\lesssim \sum_{k=1}^{n/2}k\varepsilon^{-4}e^{-\lambda(n-k)} + \sum_{k=n/2+1}^{\infty}k\varepsilon^{-4}e^{-\lambda k/2} + \sum_{k=1}^\infty k\varepsilon^{-1}e^{-\gamma_g (k+1)^2 \varepsilon^{-2}} \\[1ex]
    &\qquad\lesssim \varepsilon^{-4}\bigg(\int_1^{n/2} ke^{-\lambda(n-k)} \;dk+ \int_{n/2+1}^\infty ke^{-\lambda k/2}\;dk \bigg)
    +\varepsilon^{-1} \int_{1}^\infty ke^{-\gamma_g(k+1)^2\varepsilon^{-2}}\; dk\\[1ex]
    &\qquad\sim \varepsilon^{-4} (r+1)e ^{-\lambda r} + \varepsilon e^{-\gamma_g \varepsilon^{-2}},
\end{align*}
where we use $n\sim r$. Combining \eqref{I3}-\eqref{I3_2}, we obtain
\[
I_3\lesssim \big(\#\Omega_{0}^*(\tK) + C_\tau^2(r+1)\big) \varepsilon^{-4}e ^{-\lambda r} + \big(\#\Omega_{0}^*(\tK) + \varepsilon^2C_\tau^2\big)\varepsilon^{-1} e^{-\gamma_g \varepsilon^{-2}}.
\]


\normalem
\bibliographystyle{plain}
\bibliography{main}

\end{document}